\documentclass[a4paper,11pt]{amsart}
\usepackage{amssymb,amsmath}
\usepackage{amsfonts}
\usepackage[latin1]{inputenc}
\usepackage{fancyhdr}
\usepackage{yfonts}
\usepackage{mathrsfs}
\usepackage[dvips]{graphicx}
\usepackage[rflt]{floatflt}
\usepackage{cite}
\newtheorem{prop}{Proposition}
\newtheorem{thm}{Theorem}

\newtheorem{lem}{Lemma}
\providecommand{\prt}[1]{\left( #1 \right)}
\providecommand{\e}{\phantom{aaa}}

\providecommand{\abs[1]}{\left|#1\right|}

\providecommand{\quotes[1]}{``#1"}
\providecommand{\NDofthmprop}{\\}
\usepackage{epsfig}
\title{\LARGE \bf
Rigidity and persistence for ensuring shape maintenance of
multiagent meta formations (ext'd version)}
\author{Julien M. Hendrickx, Changbin Yu, Bar\i\c{s}
Fidan and Brian D.O. Anderson}
\thanks{J.\ M.\ Hendrickx is
with Department of Mathematical Engineering, Universit\'e
catholique de Louvain, Avenue Georges Lemaitre 4, B-1348
Louvain-la-Neuve, Belgium; {\tt\small
julien.hendrickx@uclouvain.be}. His work is supported by the
Belgian Programme on Interuniversity Attraction Poles initiated by
the Belgian Federal Science Policy Office, and the Concerted
Research Action (ARC) \quotes{Large Graphs and Networks} of the
French Community of Belgium. The scientific responsibility rests
with its authors.  Julien Hendrickx holds a FNRS fellowship
(Belgian Fund for Scientific Research)}\thanks{ C.\ Yu, B.\ Fidan
and B.\ Anderson are with Australian National University and
National ICT Australia, 216 Northbourne Ave, Canberra ACT 2601
Australia ; {\tt\small
brad.yu,baris.fidan,brian.anderson@anu.edu.au}. Their work is
supported by an Australian Research Council Discovery Project
Grant and by National ICT Australia, which is funded by the
Australian Government's Department of Communications, Information
Technology and the Arts and the Australian Research Council
through the Backing Australia's Ability Initiative.}
\begin{document}
\maketitle \thispagestyle{empty} \pagestyle{empty}%
\begin{abstract}
This paper treats the problem of the merging of formations, where
the underlying model of a formation is graphical. We first analyze
the rigidity and persistence of meta-formations, which are
formations obtained by connecting several rigid or persistent
formations. Persistence is a generalization to directed graphs of
the undirected notion of rigidity. In the context of moving
autonomous agent formations, persistence characterizes the
efficacy of a directed structure of unilateral distance
constraints seeking to preserve a formation shape. We derive then,
for agents evolving in a two- or three-dimensional space, the
conditions under which a set of persistent formations can be
merged into a persistent meta-formation, and give the minimal
number of interconnections needed for such a merging. We also give
conditions for a meta-formation obtained by merging several
persistent formations to be persistent.
\end{abstract}

\section{Introduction}\label{sec:intro}

Recently, significant interest has been shown on the behavior of
autonomous agent formations (groups of autonomous agents
interacting which each other)
\cite{BaillieulSuri:2003,DasFierroKumarOstrowski:2002,
ErenAndersonMorseWhitelyBelhumeur:2004_journal,
OlfatiMurray:2002,HendrickxAndersonDelvenneBlondel:2005}, and more
recently on {meta-formations, which is the name ascribed to an
interconnection of formations, generally with the individual
formations being separate}
\cite{AndersonYuFidanHendrickx:2006,WilliamsGlavaskiSamad:2004}.
By autonomous agent, we mean here any human-controlled or unmanned
vehicle {moving} by itself and {having} a local intelligence or
computing capacity, such as ground robots, air vehicles or
underwater vehicles. Many reasons such as obstacle avoidance  {and
dealing with a predator} can indeed lead a (meta-)formation to be
split into smaller formations which are later re-merged. Those
smaller formations need to be organized in such a way that they
can behave autonomously when the formation is split. Conversely,
some formations may need to be temporarily merged into a
meta-formation {to} accomplish a certain
task, this meta-formation being split afterwards.\\

{The particular property of formations and meta-formations which
we analyze here is \emph{persistence}}. This graph-theoretical
notion {which generalizes the notion of rigidity to directed
graphs} was introduced in
\cite{HendrickxAndersonDelvenneBlondel:2005} to analyze the
behavior of autonomous agent formations governed by unilateral
distance constraints: {Many applications require the shape of a
multi-agent formation to be preserved during a continuous move.}
For example, target localization by a group of unmanned airborne
vehicles (UAVs) using either angle of arrival data or time
difference of arrival information appears to be best achieved (in
the sense of minimizing localization error) when the UAVs are
located at the vertices of a regular polygon \cite{Dogancay:2007}.
Other examples of optimal placements for groups of moving sensors
can be found in \cite{MartinezBullo:2006}. This objective can be
achieved by explicitly keeping \emph{some} inter-agent distances
constant. In other words, some inter-agent distances are
explicitly maintained constant so that all the inter-agent
distances remain constant. The information structure arising from
such a system can be efficiently modelled by a graph, where agents
are abstracted by vertices and actively constrained inter-agent
distances by edges.\\

We assume here that those constraints are unilateral, i.e., that
the responsibility for maintaining a distance is not shared by the
two concerned agents but relies on only one of them. This
unilateral character can be a consequence of the technological
limitations of the autonomous agents. Some UAV's can for example
not efficiently sense objects that are behind them or have an
angular sensing range smaller than 360°
\cite{Everett:95,Borky:97,PachterHebert:2002}. Also, some of the
authors of this paper are working with agents in which optical
sensors have blind three dimensional cones. It can also be desired
to ease the trajectory control of the formation, as it allows
so-called leader-follower formations
\cite{BaillieulSuri:2003,TannerPapasKumar:2004,ErenAndersonMorseWhiteleyBelhumeur:2005}.
In such a formation, one agent (leader) is free of inter-agent
{distance} constraints and is only constrained by the desired
trajectory of the formation, and a second agent (first follower)
is responsible for only one distance constraint and can set the
relative orientation of the formation. The other agents have no
decision power and are forced by their distance constraints to
follow the two first agents.\\

This asymmetry is modelled {using} directed edges in the graph.
Intuitively, an information structure is persistent if, provided
that each agent is trying to satisfy all the distance constraints
for which it is responsible, it can do so, with all the
inter-agent distances then remaining constant, and as a result the
formation shape is preserved. A necessary but not sufficient
condition for persistence is \emph{rigidity}
\cite{HendrickxAndersonDelvenneBlondel:2005}, which intuitively
means that, provided that all the prescribed distance constraints
are satisfied during a continuous displacement, all the
inter-agent distances remain constant (These concepts of
persistence and rigidity are more formally reviewed in the next
section). {The above} notion of rigidity can also be applied to
{structural} frameworks where the vertices correspond to joints
and the edges to bars. The main difference between rigidity and
persistence is that rigidity assumes all the constraints to be
satisfied, as if they were enforced by an external {agency or
through} some mechanical properties, while persistence considers
each constraint to be the responsibility of a single agent. {As
explained in \cite{HendrickxAndersonDelvenneBlondel:2005},
persistence implies rigidity, but it also implies that the
responsibilities imposed on each agent are not inconsistent, for
there can indeed be situations where this is so, and they must be
avoided.} {Rigidity is thus} an undirected notion (not depending
on the edge directions), while persistence is a directed one.
{Both rigidity and persistence can be analyzed from a
graph-theoretical point of view, and it can be proved
\cite{TayWhiteley:85,HendrickxAndersonDelvenneBlondel:2005,
YuHendrickxFidanAndersonBlondel:2005} that if} a formation is
rigid (resp. persistent), then almost all formations represented
by the same graph are rigid
(resp. persistent).\\

{As stated in \cite{AndersonYuFidanHendrickx:2006}, the problem of
merging rigid formations into a rigid meta-formation has been
considered in a number of places. In
\cite{Tay:84,Moukarzel:96}, the rigidity of a multi-graph 
{{(a graph in which some vertices are abstractions of smaller
graphs)} is analyzed. {In two dimensions,} the vertices of a
multi-graph can be thought as two dimensional solid bodies at the
{boundary} of which some bars can be attached; two vertices are
then connected by an edge if the corresponding bodies are attached
to the same bar. {Of course, the idea extends obviously to three
dimensions.} Operational ways to merge two rigid formations into a
larger rigid formation can also be found in
\cite{ErenAndersonMorseWhitelyBelhumeur:2004_journal,
YuFidanAnderson:2006_ACC}}.\\

In this paper, we treat the problem of determining whether a given
meta-formation obtained by merging several persistent formations
is persistent. For this purpose, we first consider the above
mentioned problem of determining whether a meta-formation obtained
by merging rigid formations is rigid. We also analyze the
conditions under which a collection of persistent formations can
be merged into a persistent meta-formation. Conditions are then
given on the minimal number of additional links that are needed to
achieve such a merging. Note that throughout all the paper, we
always assume that the internal structure of the formations cannot
be modified. {Moreover, we use a convenient graph theoretical
formalism, abstracting agents by vertices and (unilateral)
distance constraints by (directed) edges.} \\

After reviewing some properties of rigidity and persistence of
graphs in Section \ref{sec:recall}, we examine in Section
\ref{sec:merge2D} the issues mentioned above for agents evolving
in a two-dimensional space. We show in Section \ref{sec:merge3D}
how our results can be generalized in a three-dimensional space,
and explain why this generalization can only partially be
achieved. Note that some proofs are omitted for three-dimensional
space when they are direct generalization of results on
two-dimensional space. The paper ends {with the
concluding remarks in Section \ref{sec:concl}.}\\

This paper is an extended version of
\cite{HendrickxYuFidanAnderson:2008_AJC} in which some proofs are
omitted for space reasons. Some preliminary results have also been
published in \cite{HendrickxYuFidanAnderson:2006_cdc} without
proofs, and are included here at a greater level of details.
Moreover, Propositions \ref{prop:3D_imposs_merge_second} and
\ref{prop:mergetwo3D} correct the unproven Proposition 5 in
\cite{HendrickxYuFidanAnderson:2006_cdc}, which did not take the
case described in Proposition \ref{prop:3D_imposs_merge_second}
into account.

\section{{Review of } Rigidity and Persistence}\label{sec:recall}

\subsection{Rigidity}

\begin{figure}
\centering
\begin{tabular}{ccc}
\includegraphics[scale = .4]{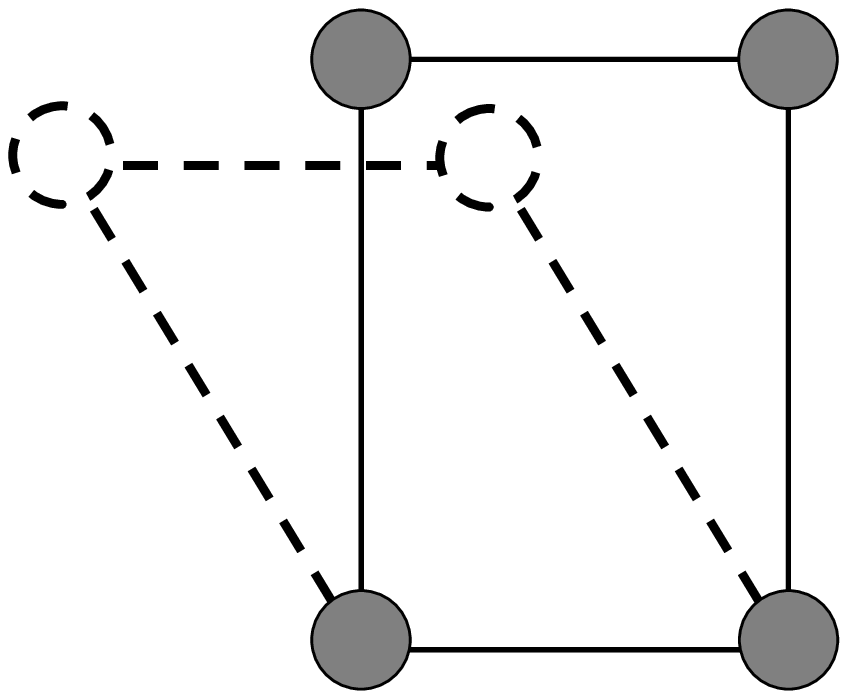}&
\includegraphics[scale = .4]{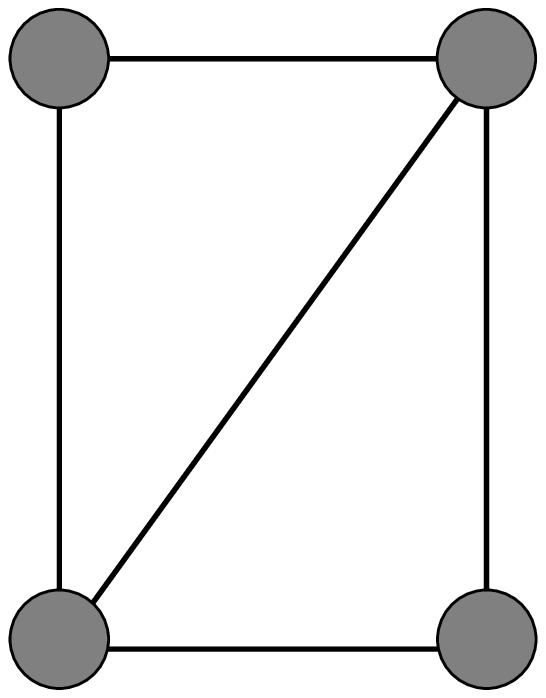}&
\includegraphics[scale = .25]{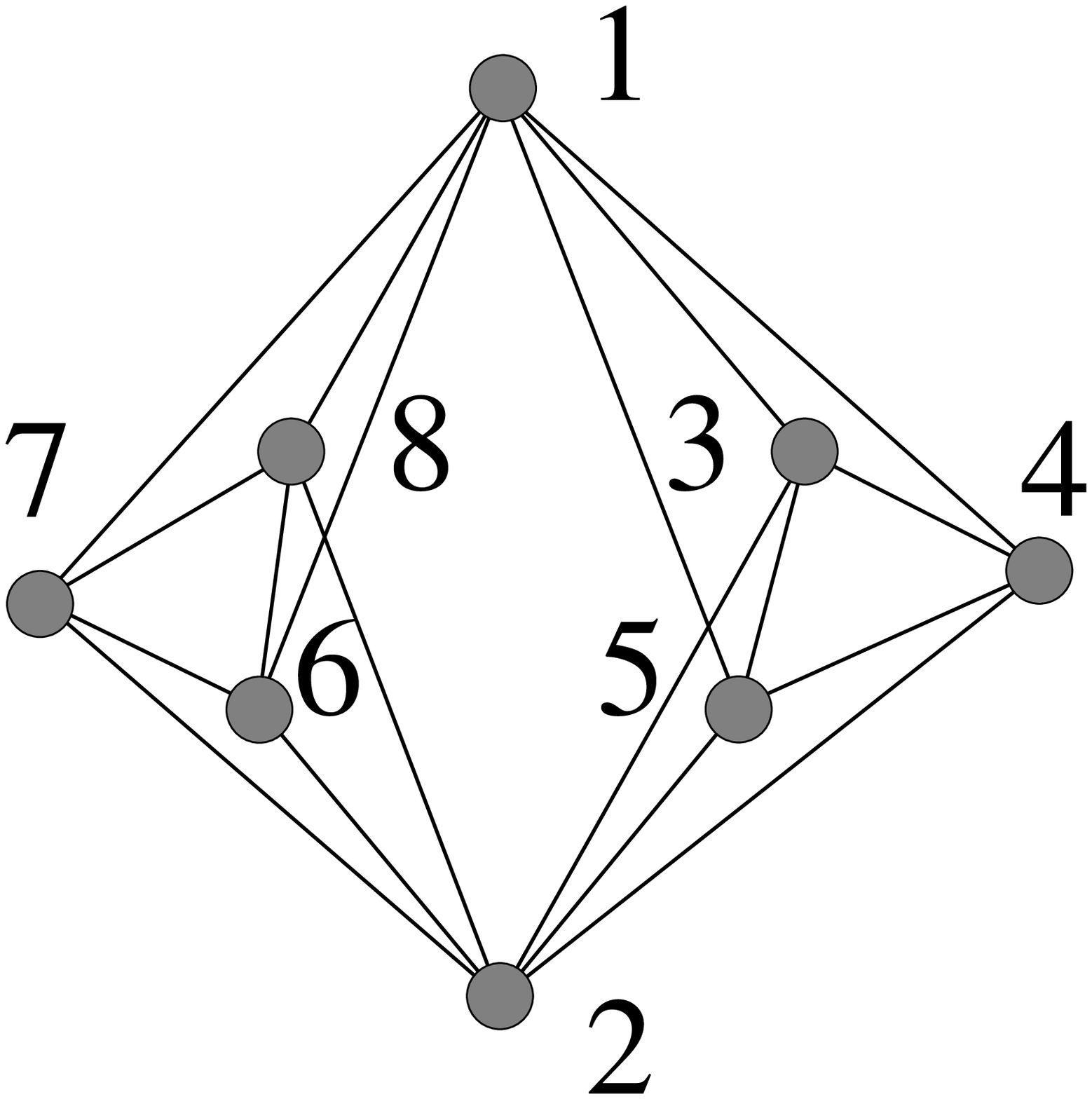}\\
(a)&(b)&(c)
\end{tabular}
\caption{In $\Re^2$, the graph represented in (a) is not rigid
because it can be deformed (dashed line), while the one in (b) is
rigid. The graph (c) satisfies the {first two conditions of
Theorem \ref{thm:Laman3D} but not the third one, and is therefore
not rigid in $\Re^3$:} the two parts of the graph can rotate
around the axis defined by 1 and 2.}\label{fig:rigidity}
\end{figure}

As explained in {Section \ref{sec:intro}}, the rigidity of a graph
has the following intuitive meaning: Suppose that each vertex
represents an agent in a formation, and each edge represents an
inter-agent distance constraint enforced by an external observer.
The graph is rigid if for almost every such structure, the only
possible continuous moves are those which preserve every
inter-agent distance, as shown in Fig. \ref{fig:rigidity}(a) and
(b). For a more formal definition, the reader is referred to
\cite{TayWhiteley:85,HendrickxAndersonDelvenneBlondel:2005}. In
$\Re^2$, that is, if the agents represented by the vertices of the
graph evolve {in two dimensions}, there exists a combinatorial
criterion to check if a given graph is rigid:
\begin{thm}[Laman \cite{Laman:70,Whiteley:96}]\label{thm:Laman}
A graph $G = (V,E)$, with $\abs{V}>1$, is rigid in $\Re^2$ if and
only if there is a sub-set $E' \subseteq E$ such that\\
(i) $\abs{E'}= 2\abs{V}-3$.\\
(ii) For all {non-empty} $E''\subseteq E'$ there holds\\
\e\e$\abs{E''}\leq 2 \abs{V(E'')}-3$,\\
where $V(E'')$ is the set of vertices incident to edges of
$E''$.\NDofthmprop
\end{thm}
{Unfortunately, the analogous criterion in $\Re^3$ is only
necessary.
\begin{thm}\label{thm:Laman3D}
If a graph $G = (V,E)$, with $\abs{V}>2$, is rigid in $\Re^3$,
there exists $E' \subseteq E$ such that\\
(i) $\abs{E'}= 3\abs{V}-6$.\\
(ii) For all non-empty $E''\subseteq E'$, there holds\\
$\abs{E''}\leq 3 \abs{V(E'')}-6$, where $V(E'')$ is the set
of vertices incident to edges of $E''$.\\
(iii) The graph $G'(V,E')$ is 3-connected (i.e. remains connected
after removal of any pair of vertices).\\
\end{thm}
Condition (iii), which also implies the 3-connectivity of $G$, is
not usually stated but is {independently} necessary even if the
two first conditions are satisfied. Fig. \ref{fig:rigidity}(c)
shows for example a non-rigid graph for which (i) and (ii) are
satisfied, but not (iii). Intuitively, the graph $G'$ in the
theorem needs to be sufficient to ensure \quotes{alone} the
rigidity of $G$. 3-connectivity is then needed as otherwise two or
more parts of the graph could rotate around the axis defined by
any pair of vertices whose removal would disconnect the graph.
Note that such connectivity condition is not necessary in
2-dimensional spaces, as the counting conditions (i) and (ii) of
Theorem \ref{thm:Laman} imply the 2-connectivity. For more
information on necessary conditions for rigidity in
three-dimensional spaces, we refer the
reader to \cite{MantlerSnoeyink:2004}}.\\

We say that a graph is \emph{minimally rigid} if it is rigid and
if no single edge can be removed without losing rigidity. {It
follows from the results above that} a graph is minimally rigid in
$\Re^2$ (resp. in $\Re^3$) if and only if it {is rigid and}
contains $2\abs{V}-3$ (resp. $3\abs{V}-6$) edges
\cite{TayWhiteley:85}. Therefore we have the following
characterization of minimal rigidity in $\Re^2$.
\begin{thm}[Laman \cite{Laman:70,Whiteley:96}]\label{thm:minLaman}
A graph $G = (V,E)$, with $\abs{V}>1$, is minimally rigid in
$\Re^2$ if and only if it is rigid and contains $2\abs{V}-3$
edges, or equivalently if and only if\\
(i) $\abs{E}= 2\abs{V}-3$.\\
(ii) For all {non-empty} $E''\subseteq E$ there holds\\
\e\e$\abs{E''}\leq 2 \abs{V(E'')}-3$, where $V(E'')$ is the set of
vertices incident to edges of $E''$.\NDofthmprop
\end{thm}

The notion of rigidity can also be described from a linear
algebraic point of view, using the so-called rigidity matrix.
Suppose that a position $p_i \in \Re^d$ (with $d=2,3$) is given to
each vertex $i$ of a graph $G=(V,E)$, and let $p\in
\Re^{d\abs{V}}$ be the juxtapositions of all positions. For each
vertex, consider now an infinitesimal displacement $\delta p_i$,
and let $\delta p$ be a vector obtained by juxtaposing these
displacements. {Since with infinitesimal displacements one can
neglect higher order terms,} the distance between the positions of
two vertices $i$ and $j$ is preserved by the set of infinitesimal
displacements if
\begin{equation}\label{eq:linconstr}
\prt{p_i-p_j}^T\prt{\delta p_i-\delta p_j}=0.
\end{equation}
Hence, if each edge represents a distance constraint, a set of
infinitesimal displacements is allowed if and only if
(\ref{eq:linconstr}) is satisfied for any edge $(i,j)\in E$. This
set of linear constraints can be conveniently re-expressed in a
condensed form as $R_G \delta p = 0$ where
$R_G\in\Re^{\abs{E}\times d\abs{V}}$ is the rigidity matrix, which
contains one row for each edge and $d$ columns for each vertex. In
the row corresponding to the edge $(i,j)$, the $d(i-1)+1^{st}$ to
$di^{th}$ columns are $(p_i-p_j)^T$, the $d(j-1)+1^{st}$ to
$dj^{th}$ columns are $(p_j-p_i)^T$, and all other columns are 0.
A graph $G$ is rigid if for almost all position assignment its
rigidity matrix has a rank $d\abs{V} - f(d,\abs{V})$, {where
$f(d,\abs{V})$ is the number of degrees of freedom in a
$d-$dimensional space of a $\min(\abs{V}-1,d)$-dimensional rigid
body {(Observe that $\min(\abs{V}-1,d)$ is the largest possible
dimension of a graph on $\abs{V}$ vertices embedded in a
$d$-dimension space).} In a 2-dimensional space, a single point
has two DOFs $f(2,1)=2$, and any one or two-dimensional body has
three DOFs. In a three-dimensional space, a single point has three
DOFs, a one-dimensional object has five
DOFs, and any other object has six DOFs.}\\
A subgraph $G'(V',E')\subseteq G(V,E)$ is rigid if the restriction
$R_{G'}$ of $R_G$ to the rows and columns corresponding to $E'$
and $V'$ has a rank $d\abs{V'} - f(d,\abs{V'})$. Note that the
rank $d\abs{V'} - f(d,\abs{V'})$ is the maximal that can be
attained by a rigidity (sub-)matrix. In a minimally rigid
(sub-)graph, this rank is attained with a minimal number of edges
and all rows of the rigidity matrix are thus linearly independent.
For more information on the rigidity matrix, we refer the reader
to \cite{TayWhiteley:85}.\\

\subsection{Persistence}

Consider now that the constraints are not enforced by an external
{entity, but that each constraint is the responsibility of one
agent to enforce.} To each agent, one assigns a (possibly empty)
set of unilateral distance constraints represented by directed
edges: the notation $\overrightarrow{(i,j)}$ for a directed edge
connotes that the agent $i$ has to maintain its distance to $j$
constant during any continuous move. As explained in the
Introduction, the \emph{persistence} of the directed graph means
that provided that each agent is trying to satisfy its
constraints, the distance between any pair of connected or
non-connected agents is maintained constant during any continuous
move, and as a consequence the shape of the formation is
preserved. {Note though that the {assignments} given to an agent
may be impossible to fulfill, in which case persistence is not
achieved.} An example of a persistent and a non-persistent graph
having the same underlying undirected graph {is} shown in Fig.
\ref{fig:persistence}. For a more formal definition of
persistence, the reader is referred to
\cite{HendrickxAndersonDelvenneBlondel:2005,YuHendrickxFidanAndersonBlondel:2005},
where are also proved the rigidity of all persistent graphs and
the following criterion to check persistence:
\begin{thm}\label{thm:crit_persi}
A graph $G$ is persistent in $\Re^2$ (resp. $\Re^3$) if and only
if every subgraph obtained from $G$ by removing edges leaving
vertices {whose} out-degree is greater than 2 (resp. 3) until no
such
vertex is present anymore in the graph is rigid.\\
\end{thm}
A key result in the proof of Theorem \ref{thm:crit_persi}
\cite{HendrickxAndersonDelvenneBlondel:2005,YuHendrickxFidanAndersonBlondel:2005}
is the following:
\begin{prop}\label{prop:d+}
A persistent graph $\Re^2$ (resp. $\Re^3$) remains persistent
after removal of an edge leaving a vertex {whose} out-degree is
larger than 2 (resp. 3).\NDofthmprop
\end{prop}

\begin{figure}
\centering
\begin{tabular}{cc}
\includegraphics[scale = .4]{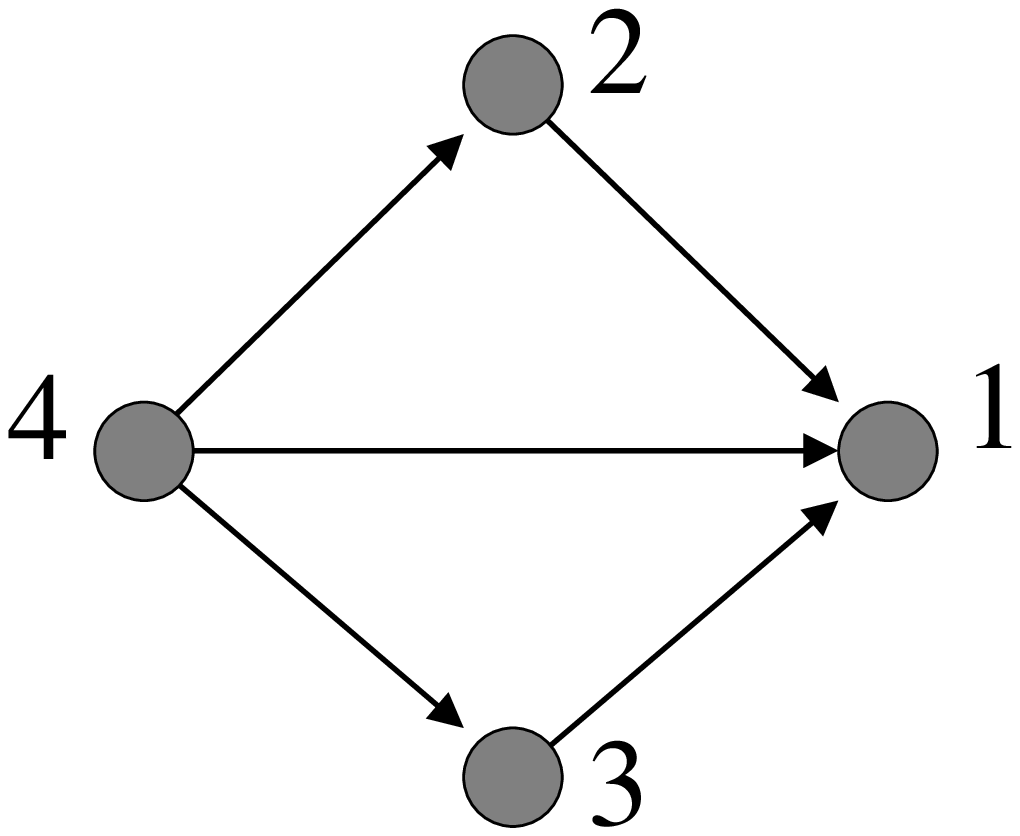}&
\includegraphics[scale = .4]{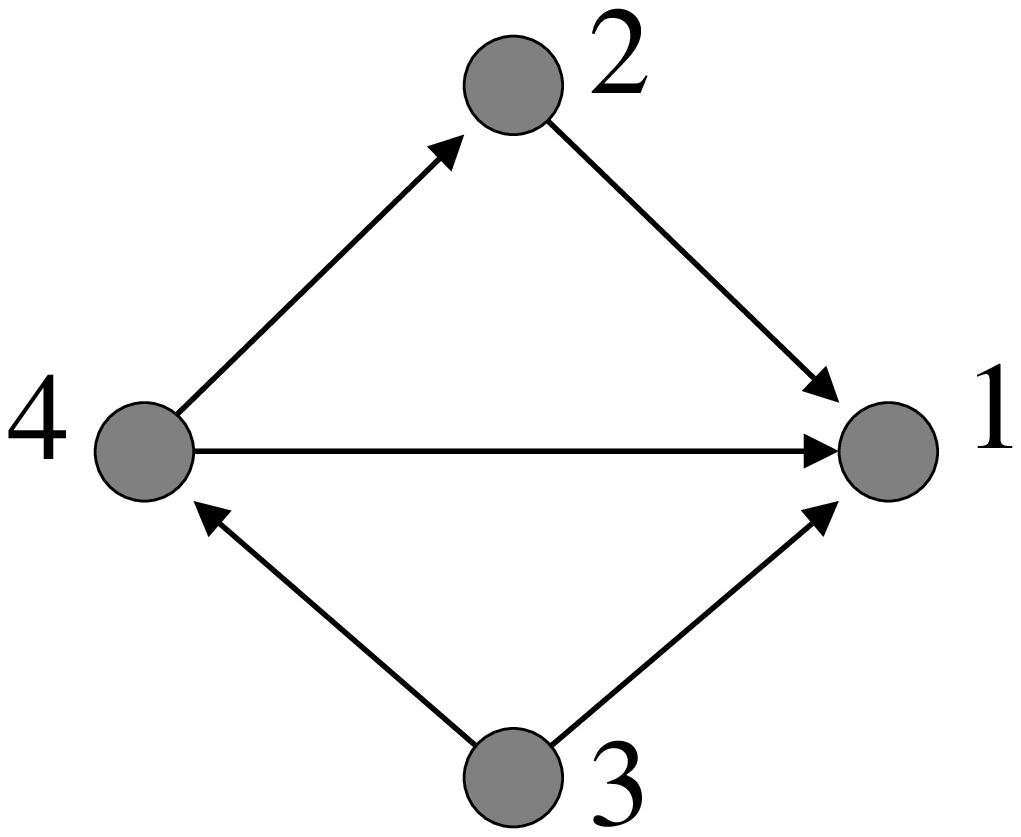}\\
(a)&(b)
\end{tabular}
\caption{In $\Re^2$, the graph represented in (a) is rigid but not
persistent. For almost all uncoordinated displacements of 2, 3 and
4 (even if they satisfy their constraints), 4 is indeed unable to
satisfy its three constraints. This problem cannot happen for the
graph represented in (b), which is persistent.
}\label{fig:persistence}
\end{figure}

We {use the term} \emph{number of degrees of freedom} of a vertex
$i$ {to denote} the (generic) dimension of the set in which the
corresponding agent can choose its position (all the other agents
being fixed). {Thus it} represents in some sense the decision
power of this agent. In a three-dimensional space, an agent being
responsible for one distance constraint can for example freely
move on the surface of a sphere centered on the agent from which
the distance needs to be maintained, and has thus two degrees of
freedom. The number of degrees of freedom of a vertex $i$ in
$\Re^2$ (resp. $\Re^3$) is given by $\max\prt{0,2-d^+(i)}$ (resp.
$\max\prt{0,3-d^+(i)}$), where $d^+(i)$ represent the out-degree
of the vertex $i$. A vertex having a maximal number of degrees of
freedom (i.e. an out-degree 0) is called a \emph{leader} since the
corresponding agent does not have any distance constraint to
satisfy. {We call \emph{the number of degrees of freedom of a
graph} the sum of the numbers of degrees of freedom over all its
vertices. It is proved in
\cite{HendrickxAndersonDelvenneBlondel:2005,YuHendrickxFidanAndersonBlondel:2005}
that this quantity cannot exceed 3 in $\Re^2$ and 6 in $\Re^3$.}
Note that those numbers correspond to the number of {independent}
translations and rotations in $\Re^2$ and $\Re^3$. {In the sequel
we abbreviate degree of freedom by
DOF}.\\

As explained in \cite{YuHendrickxFidanAndersonBlondel:2005},
although the {concept} of persistence {is applicable} in three and
{larger} dimensions, it is not sufficient to imply the desired
stability of the formation shape. {For the shape stability, t}he
graph corresponding to a three-dimensional formation needs {in
addition} to be \emph{structurally persistent}. In $\Re^3$, a
graph is structurally persistent if and only if it is persistent
and contains at most one leader{, i.e. at most one vertex with no
outgoing edge}. In $\Re^2$,
{persistence and structural persistence} are equivalent.\\

Similarly to minimal rigidity, we say that a graph is
\emph{minimally persistent} if it is persistent and if no single
edge can be removed without losing persistence. It is proved in
\cite{HendrickxAndersonDelvenneBlondel:2005,YuHendrickxFidanAndersonBlondel:2005}
that a graph is minimally persistent if and only if it is
persistent and minimally rigid. The number of edges of such a
graph is {thus uniquely determined by the number of its vertices}
{as it is the case for minimally rigid graphs}.

\section{Rigidity and Persistence of 2D {Meta-Formations}}\label{sec:merge2D}

\subsection{Rigidity}\label{sec:2DRig}

Consider a set $N$ of disjoint rigid (in $\Re^2$) graphs
$G_1,\dots,G_{\abs{N}}$ having at least two vertices each, and a
set $S$ of single-vertex graphs
$G_{\abs{N}+1},\dots,G_{\abs{N}+\abs{S}}$. In the sequel, those
graphs are called \emph{meta-vertices}, and it is assumed that no
modification can be made {on} their internal structure: no
internal edge or vertex can be added to or removed from a
meta-vertex. We define the merged graph $G$ by taking the union of
all the meta-vertices, and of some additional edges $E_M$ {each of
which has end-points belonging to different meta-vertices}.\\

The conditions under which the merging of two meta-vertices leads
to a rigid graph are detailed in \cite{YuFidanAnderson:2006_ACC}:
If both meta-vertices contain more than one vertex, the merged
graph is rigid if and only if $E_M$ contains at least three edges,
{the aggregate of} which are incident to at least two vertices of
each meta-vertex. This is actually a particular case of the
following result for an arbitrary number of graphs ({analogous to
a result in \cite{Moukarzel:96} which is obtained under the
assumption that no vertex of any meta-vertex is incident to more
than one edge of $E_M$)}:
\begin{thm}\label{thm:metalaman2D}
{{If it contains at least two vertices,} $G=\prt{\bigcup_{N,S}
G_i}\cup E_M$ (with $N$ and $S$ as defined at the beginning of
this section)} is rigid if and only if
there exists $E_M'\subseteq E_M$ such that\\
(i) $\abs{E_M'}= 3\abs{N}+2\abs{S}-3$.\\  %
(ii) For all {non-empty} $E_M''\subseteq E_M'$, there holds\\
\e\e$\abs{E_M''}\leq 3\abs{I(E_M'')}+2\abs{J(E_M'')}-3$,\\
{where $I(E_M'')$ is the set of meta-vertices such that there are
at least two vertices within the meta-vertex all incident to edges
of $E_M''$, and $J(E_M'')$ is the set of meta-vertices such that
there is precisely one vertex within the meta-vertex that is
incident to one or several edges of $E_M''$. Note that in each
case, there can be an arbitrary number of vertices in the
meta-vertex which are not incident on any edge of
$E_M''$.}\NDofthmprop
\end{thm}
{To prove this theorem, we first need the following lemma, which
we shall prove the both for $\Re^2$ and $\Re^3$, intending to use
the $\Re^3$ result in the next section.}
\begin{lem}\label{lem:lem2metalaman}
Let $G(V,E)$ be a rigid graph (in $\Re^2$ or $\Re^3$), and
$G_1',\dots G_N'$ be minimally rigid subgraphs of $G$ having
distinct vertices. Then there exists a minimally rigid subgraph
$G'(V,E')$ of $G$ containing all vertices of $G$ and all subgraphs
$G_i$.
\end{lem}
\begin{proof}
{For simplicity, let us first consider the 2-dimensional case.}
Consider the rigidity matrix $R_G$ of $G$. Since $G$ is rigid, it
has (for almost all positions) a rank $2\abs{V}-3$. Since each
$G_i'$ is minimally rigid, the restriction $R_{G_i'}$ of $R_G$ to
the rows and columns corresponding to the edges and vertices of
$G_i'$ has $2\abs{V_i}-3$ linearly independent rows (or is an
empty matrix if $\abs{V_i}=1$). Also, since the vertices of the
different $G_i'$ are distinct, there can be no dependence between
rows corresponding to edges of different subgraphs $G_i'$.
Therefore, all rows of $R_{\bigcup G_i'}$, corresponding to all
edges of $\bigcup G_i'$, are linearly independent. Since the rank
of $R_G$ is $2\abs{V}-3$, it is a standard result in linear
algebra that $R_{\bigcup G_i'}$ can be completed by the addition
of further rows of $R_G$ to obtain a subset of $2\abs{V}-3$
linearly independent rows of $R_G$. Letting $E'$ be the set of
edges corresponding to this set of rows, the graph $G'(V,E')$ is a
minimally rigid subgraph of $G$ containing all $G_i'$. This
completes the proof for the 2-dimensional case. {The proof for the
3-dimensional case is established following the same steps above,
but replacing $2\abs{V}-3$ by $3\abs{V}-6$ and adding a special
case for $\abs{V_i}=2$ in addition to the case where
$\abs{V_i}=1$.}\\
\end{proof}
We can now prove Theorem \ref{thm:metalaman2D}.
\begin{proof}
For every $G_i$, let $G_i'$ be a minimally rigid subgraph of $G_i$
on the same vertices (The existence of such subgraphs follows
directly from the definition of minimal rigidity, and they can be
obtained by successively removing edges from the initial graph).
Since they are minimally rigid, they contain $2\abs{V_i}-3$ edges
if $G_i\subseteq N$ and no edge if $G_i\subseteq S$.\\

We first suppose that there exists a set $E_M'$ as described in
the theorem and prove the rigidity of $G$, by proving the minimal
rigidity of one of its subgraph viz.,  $G' = (V,E') =
\prt{\bigcup_{N,S} G_i'}\cup E_M'$ which contains all its
vertices. The number of edges in $G'$ is
\begin{equation*}
\begin{array}{ll}
\abs{E'} &= \abs{E_M'} + \sum_{G_i\in N} \abs{E_i'}\\ &=
3\abs{N}+2\abs{S}-3 + \sum_{G_i\in N} (2\abs{V_i}-3)
\\&= 2\abs{V}-3,
\end{array}
\end{equation*}
since $\abs{V}=\abs{S } + \sum_{G_i \in N} \abs{V_i}$. To show
that $G'$ satisfies the second condition of Theorem
\ref{thm:Laman}, suppose that there exists a subset of edges
$E''\subset E'$ such that $\abs{E''}>2\abs{V(E'')}-3$, let $I$ be
the set of meta-vertices containing at least two vertices of
$V(E'')$ and $J$ the set of meta-vertices containing only one
vertex of $V(E'')$. Let now $E_M''=E_M\cap E''$ and for each $i$,
$V_i'' = V(E'') \cap V_i$ and $E_i'' = E'' \cap E_i''$. There
holds $V(E'') = \sum_{G_i\in I}\abs{V_i''} + \abs{J}$, and $E'' =
E_M'' + \sum_{G_i\in I}\abs{E_i''}$. Moreover, since each $G_i'$
is minimally rigid, it follows from Theorem \ref{thm:minLaman}
that $\abs{E_i''}\leq 2\abs{V_i''}-3$. We have then
\begin{equation*}
\begin{array}{ll}
\abs{E_M''} &= \abs{E''} - \sum_{G_i\in I}\abs{E_i''}\\&
>2\abs{V''}-3 - \sum_{G_i\in I}\prt{2\abs{V_i''}-3} \\&= 3\abs{I} +
2\abs{J}  -3,
\end{array}
\end{equation*}
so that this $E_M''\subseteq E_M'$ does not satisfy condition (ii)
in the theorem.\\

We now suppose that $G$ is rigid. It follows from Lemma
\ref{lem:lem2metalaman} that there is a minimally rigid subgraph
$G'(V,E')\subseteq G$ containing all $G_i'$. Let $E_M'=E'\cap
E_M$; we prove that $E_M'$ satisfies the condition of this
theorem. Since $G'$ is minimally rigid, there holds
$\abs{E'}=2\abs{V}-3$. Moreover, we have $\abs{E'}= \abs{E_M'} +
\sum_{i\in N} \abs{E_i'}$, and $\abs{V}  = \sum_{G_i\in N}
\abs{V_i'}+ \abs{S}$, so that
\begin{equation*}
\abs{E_M'} = 2\abs{V}-3 - \sum_{G_i\in N} (2\abs{V_i}-3) =
3\abs{N} + 2\abs{S}-3.
\end{equation*}
$E_M'$ contains thus the predicted number of edges. We suppose now
that there is a set $E_M''$ such that $\abs{E_M''}>
3\abs{I(E_M'')}+2\abs{J(E_M)}-3$ and show that this contradicts
the minimal rigidity of $G'$. Let us build $E''$ by taking the
union of $E_M''$ and all $E_i'$ for which $G_i\in I(E_M'')$. There
holds $\abs{V(E'')}= \abs{J(E_M'')} + \sum_{G_i\in
I(E_M'')}\abs{V_i}$. Therefore, we have
\begin{equation*}
\begin{array}{lll}
\abs{E''}& = & \abs{E_M''} + \sum_{i\in I(E_M'')}\abs{E_i'}\\& >&
3\abs{I(E_M'')}+2\abs{J(E_M'')}-3 \\&&+ \sum_{G_i\in
I(E_M'')}\prt{2\abs{V_i}-3}\\& =& 2V(E'') -3.
\end{array}
\end{equation*}
By Theorem \ref{thm:minLaman}, this contradicts the
minimal rigidity of $G'(V,E')$ as $E''\subseteq E'$.\\
\end{proof}

This criterion can be checked by a quadratic time algorithm (with
respect to the number of meta-vertices) {which would be a simple
adaptation of the pebble game algorithm that is used for rigid
graphs (see \cite{JacobsHendrickson:97})}, or even faster
\cite{Moukarzel:96}.\\

For a given collection of meta-vertices, we say that $G$ is an
\emph{{edge}-optimal rigid merging} if no single edge of $E_M$ can
be removed without losing rigidity. Notice that a single graph can
be an {edge}-optimal rigid merging with respect to a certain
collection of meta-vertices, and not with respect to another one,
as shown in Fig. \ref{fig:opt_and_notopt_rigmerging}. If all
meta-vertices are minimally rigid, then an {edge}-optimal rigid
merging is also a minimally rigid graph. From Theorem
\ref{thm:metalaman2D}, one can deduce the following
characterization of {edge}-optimal rigid merging.

\begin{thm}\label{thm:edge-optimal_rig2d}
{$G=\prt{\bigcup_{N,S} G_i}\cup E_M$ (with $N$ and $S$ as defined
at the beginning of this section)} {containing at least two
vertices} is an edge-optimal rigid merging if and only if it is
rigid and satisfies $\abs{E_M}= 3\abs{N}+2\abs{S}-3$. Moreover,
each rigid merging contains an edge-optimal rigid merging on the
same set of meta-vertices.
\end{thm}
\begin{proof}
{Observe first that Theorem \ref{thm:metalaman2D} requires a rigid
merged graph $G$ to satisfy $E_M\geq 3\abs{N}+2\abs{S}-3$.
Therefore a rigid merged graph for which $E_M=
3\abs{N}+2\abs{S}-3$ is an edge-optimal merging. Let now $G$ be a
rigid merged graph. By Theorem \ref{thm:metalaman2D} there exists
$E_M'\subset E_M$ with $E_M'= 3\abs{N}+2\abs{S}-3$ satisfying
condition (ii) of this same theorem. One can see, again using
Theorem \ref{thm:metalaman2D}, that $G'=\prt{\bigcup_{N,S}
G_i}\cup E_M'$ is rigid, as the set $E_M'$ trivially contains
itself and satisfies both conditions (i) and (ii). It follows then
from the size of $E_M'$ and from the discussion above that $G'$ is
an edge-optimal rigid merging. We have thus proved that any rigid
merged graph $G$ contains an edge-optimal rigid merged graph $G'$
on the same meta-vertices satisfying $E_M'= 3\abs{N}+2\abs{S}-3$.
Therefore it cannot contain less than $3\abs{N}+2\abs{S}-3$ edges,
and if it contains more of them, it is not edge-optimal. It is
thus edge-optimal if and only if $E_M= 3\abs{N}+2\abs{S}-3$. }
\end{proof}

\begin{figure}
\centering
\begin{tabular}{cc}
\includegraphics[scale = .3]{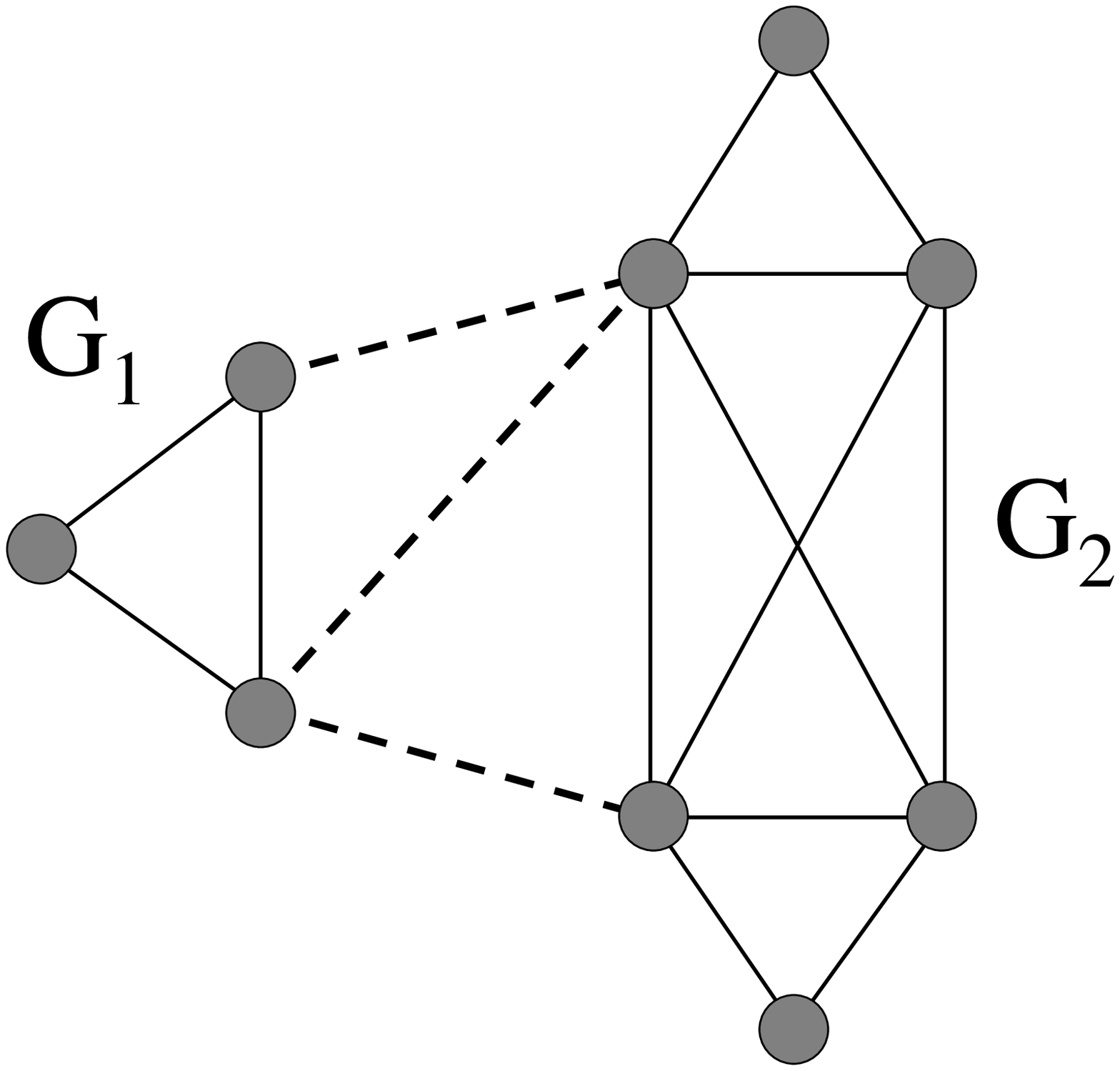}&
\includegraphics[scale = .3]{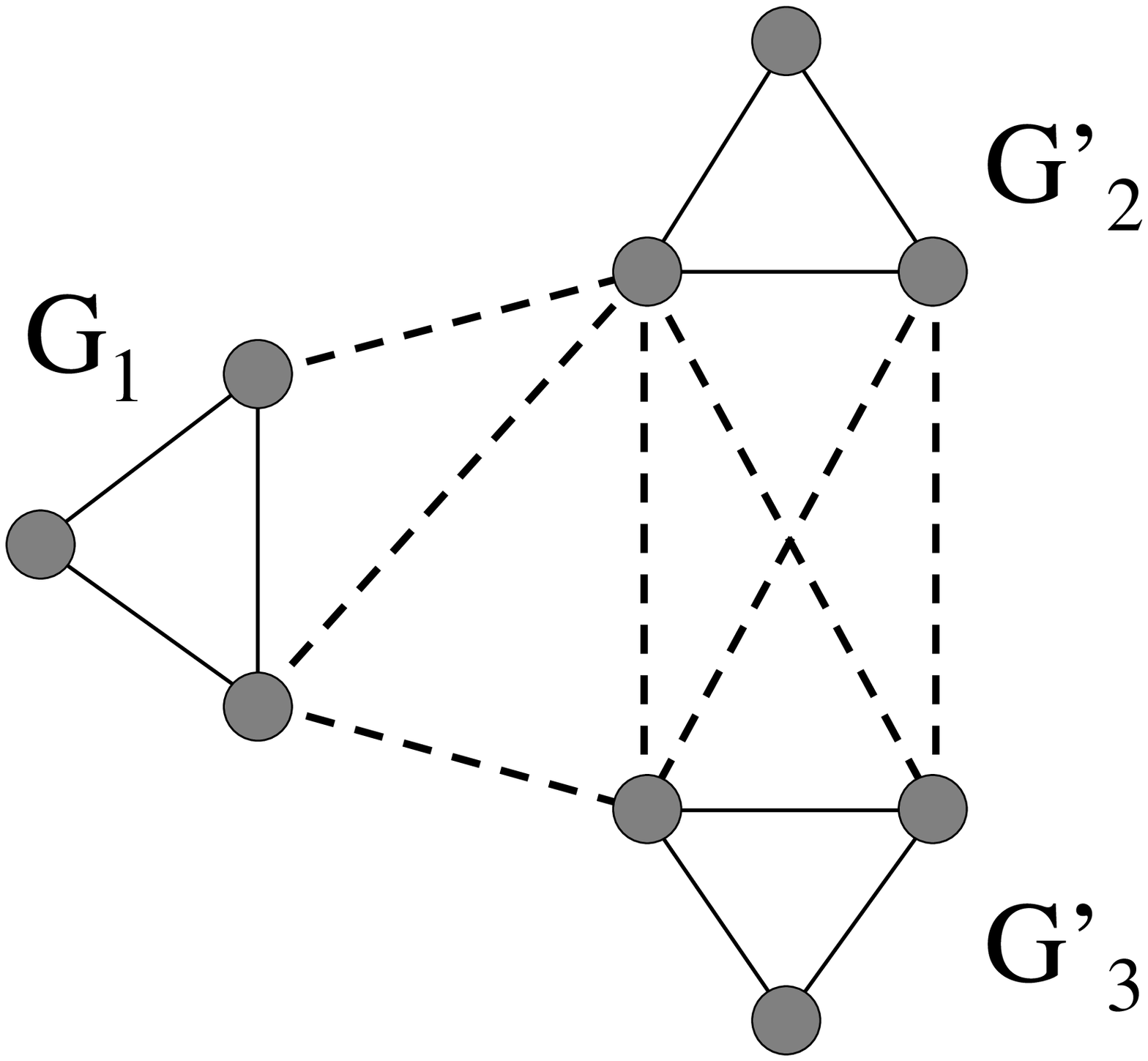}\\
(a) &(b)
\end{tabular}
\caption{The graph represented in (a) and (b) is an {edge}-optimal
rigid merge if it is obtained by merging $G_1$ and $G_2$ (a) but
not if it is obtained by merging $G_1$, $G_2'$ and $G_3'$ (b). The
dashed edges represent the edges of
$E_M$}\label{fig:opt_and_notopt_rigmerging}
\end{figure}

\subsection{Persistence}\label{sec:2DPer}

{{Next we analyze the case where} the meta-vertices $G_i$ are
directed persistent graphs, and adapt the definitions of $N$ and
$S$ in consequence.} If it is possible to merge them into a
persistent graph, then it is possible to do so in such a way that
all the edges of $E_M$ leave vertices which have an out-degree not
greater than 2 in $G$: a set of edges $E_M$ that would make $G$
persistent but that would not satisfy this property could indeed
be reduced by Proposition \ref{prop:d+} until it satisfies it.
{Moreover, we have the following proposition.}
\begin{prop}\label{prop:necsufper_when_em<2}
Let $G=\prt{\bigcup_{N,S} G_i}\cup E_M$ with $N$ and $S$ as
defined at the beginning of this section, and with all $G_i$
persistent. If no vertex left by an edge of $E_M$ has an
out-degree larger than 2, then $G$ is persistent if and only if it
is rigid.
\end{prop}
\begin{proof}
Rigidity is a necessary condition for persistence, so we just have
to prove that it is here sufficient. Let $G'$ be a (directed)
graph obtained from $G$ by removing edges leaving vertices with
out-degree larger than 2 until no such vertex exists in the graph.
It follows from Theorem \ref{thm:crit_persi} that we just need to
prove the rigidity of any such $G'$. For every $i$, let $G_i'$ be
the restriction of $G'$ to the meta-vertex $G_i$. Since in $G$,
every edge of $E_M$ leaves a vertex with an out-degree at most 2,
there holds $G'=\prt{\bigcup_{N,S} G_i'}\cup E_M$ as no edge of
$E_M$ is removed when building $G'$. Moreover, for every $i$,
$G_i'$ can be obtained from $G_i$ by removing edges leaving
vertices with an out-degree larger than 2 until no such vertex
exists in the graph anymore. The only vertices that are not left
by exactly the same edges in $G$ as in $G_i$ are indeed those left
by edges of $E_M$, which by hypothesis have an out-degree at most
2 and are therefore unaffected by the edge-removal procedure. It
follows then from the persistence of all $G_i$ and from Theorem
\ref{thm:crit_persi} that all $G_i'$ are rigid. And since $G$ is
rigid, $E_M$ satisfies the necessary and sufficient conditions of
Theorem \ref{thm:metalaman2D}. Therefore, the graph
$G'=\prt{\bigcup_{N,S} G_i'}\cup E_M$ is also rigid, as the
conditions of Theorem \ref{thm:metalaman2D} do not depend on the
edges inside the different meta-vertices. As explained above, this
implies the persistence of $G'$.\\
\end{proof}

The condition on the out-degrees of the vertices {with an outgoing
edge of $E_M$ can be} conveniently re-expressed in terms of
degrees of freedom: To each DOF ({within a single} meta-vertex) of
any vertex {there} corresponds at most one outgoing edge of $E_M$.
By an abuse of language, we say that such edges leave {a vertex
with one or more} \emph{local} DOFs, {i.e. a vertex which inside
its meta-vertex has one or more DOFs and {\emph{which is then left
by no more edges of $E_M$ than the number of DOFs is has.}}} This
allows reformulating Proposition \ref{prop:necsufper_when_em<2},
{the proof of which can directly be extended to any dimension}, in
a dimension-free way:
\begin{thm}\label{thm:leavedof2}
A collection of persistent meta-vertices can be merged into a
persistent graph if and only if it can be merged into a persistent
graph by adding edges leaving {vertices with one or more} local
DOFs, {the number of added edges not exceeding the number of local
DOFs.} In that case, the merged graph is persistent if and only if
it is rigid.\NDofthmprop
\end{thm}

If {one or more edges} of $E_M$ do leave a vertex with an
out-degree larger than 2, no criterion {has been found yet} {to
determine whether the merged graph is persistent or not,} {while
also taking} {advantage of the fact that the graph is obtained by
merging several persistent meta-vertices.}
\\

{Tying Theorem \ref{thm:leavedof2} together with what is known and
reviewed above regarding the merging of two rigid meta-vertices,
we conclude:} two persistent meta-vertices $G_a$ and $G_b$ {each}
having two or more vertices can be merged into a persistent graph
if and only if three edges leaving {vertices with} local DOFs can
be added in such a way that they are incident to at least two
vertices in each meta-vertex. There must thus be at least three
{local} DOFs available among the {vertices in} $G_a$ and $G_b$.
Conversely, if {there are available three local DOFs among the
vertices of $G_a$ and $G_b$}, since no vertex can have more than
two DOFs, it is possible to add a total of at least three edges
leaving at least two vertices of $G_a \cup G_b$. The vertices to
which those edges arrive can then be chosen in such a way that at
least two vertices of both $G_a$ and $G_b$ are incident to edges
of $E_M$, as in the example shown in Fig. \ref{fig:merge_two_2D}.
It follows then from Theorem \ref{thm:metalaman2D} that this graph
is rigid, which by Theorem \ref{thm:leavedof2} implies that the
merged graph is persistent:
\begin{prop}\label{prop:mergetwo2D}
Two persistent meta-vertices each having two or more vertices can
be merged into a persistent graph if and only if the sum of their
DOF numbers is at least 3. At least three edges are needed to
perform this merging, {and merging} can always be done with
exactly three edges.\NDofthmprop
\end{prop}

\begin{figure}
\centering
\begin{tabular}{ccc}
\includegraphics[scale = .35]{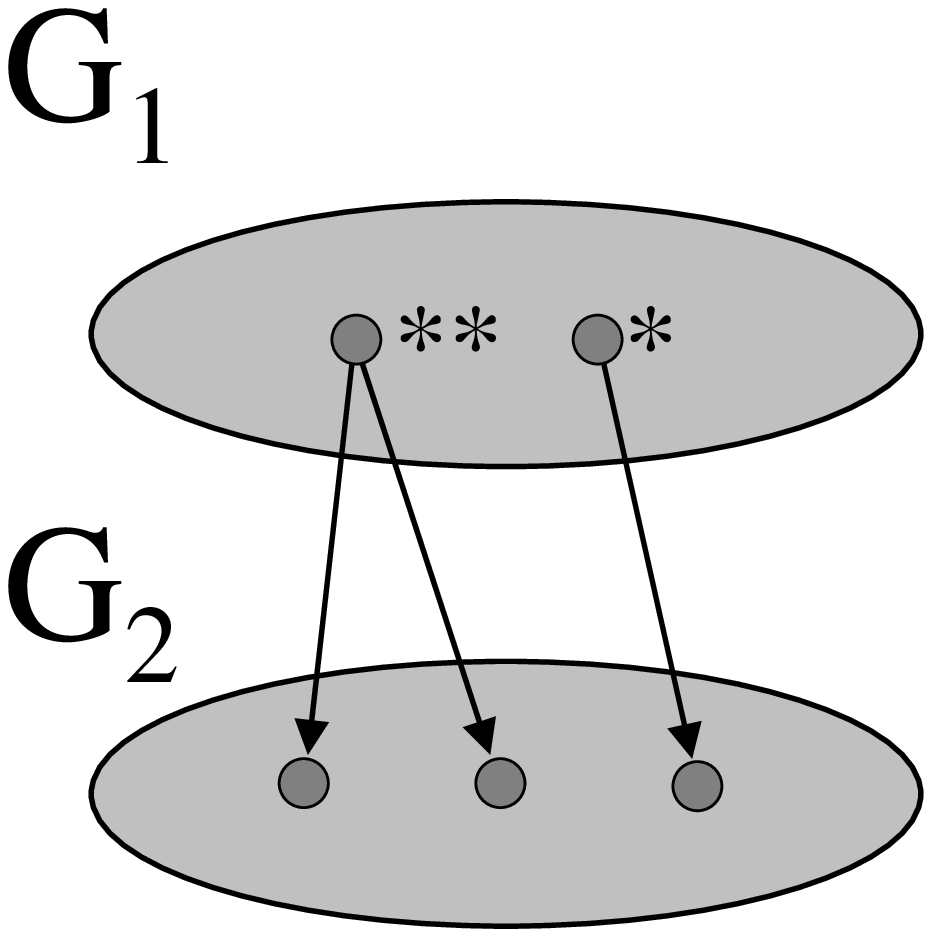}&&
\includegraphics[scale = .4]{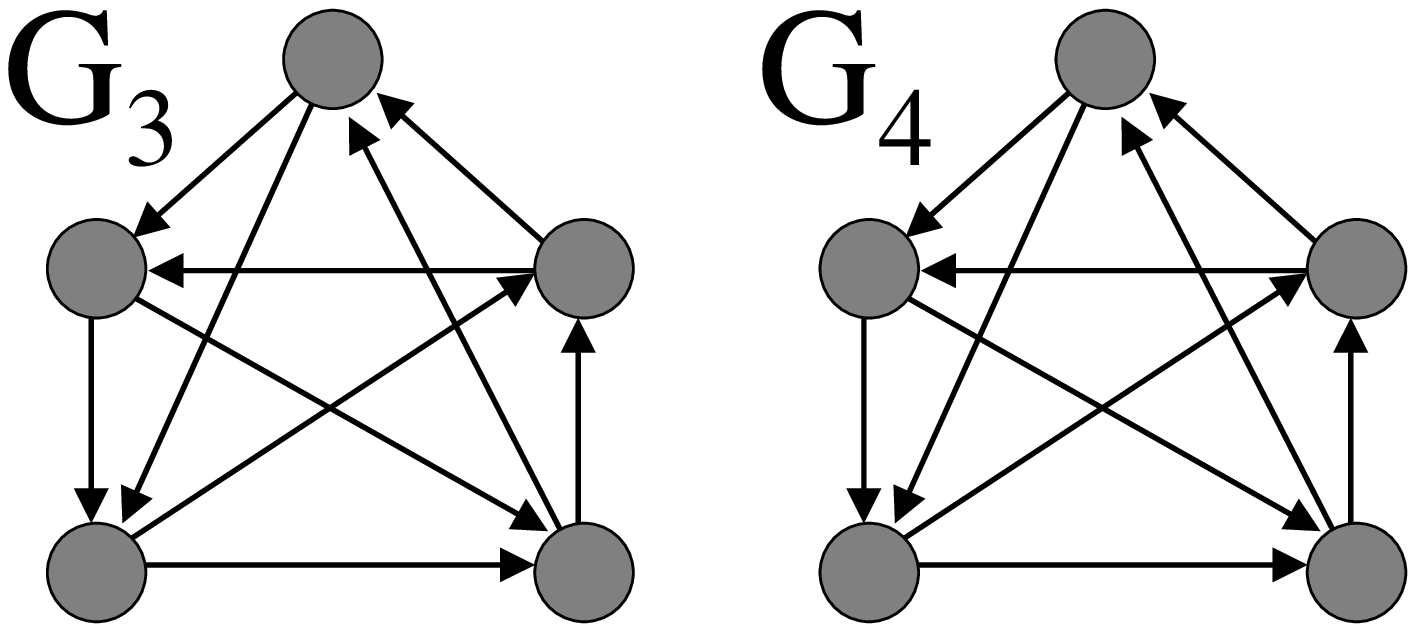}\\
(a)&&(b)
\end{tabular}

\caption{Merging of the persistent meta-vertices $G_1$ and $G_2$
into a persistent graph in $\Re^2$ (a). The symbol \quotes{*}
represents one DOF (with respect to the meta-vertex). (b)
represents two persistent meta-vertex that cannot be merged into a
persistent graph in $\Re^2$ {by addition of interconnecting edges}
because none of their vertices has a DOF.}\label{fig:merge_two_2D}
\end{figure}

If one or two of the meta-vertices are single vertex graphs, the
result still holds, but the minimal number of added edges (and
therefore the number of needed DOFs) are then respectively 2 and
1. We {define the} \emph{number of missing DOFs} ($m_{DOF}$) {to
be} {the maximal number of DOFs that any graph with the same
number of vertices can have, less the number of DOFs the graph
actually has}. In $\Re^2$, this maximal number is 2 for the single
vertex graphs, and 3 {for other {persistent} graphs}. {There is an
interesting consequence: when} the minimal number of edges is used
to merge two meta-vertices $G_a$ and $G_b$, the number of missing
DOFs is preserved through the process, i.e. $m_{DOF}\prt{G_a\cup
G_b \cup E_M} = m_{DOF}(G_a) +
m_{DOF}(G_b)$.\\

Consider now an arbitrary number of persistent meta-vertices,
possibly containing single-vertex graphs, but such that the total
number of vertices is at least 2. If the sum of their number of
missing DOFs is no greater than 3, it follows from Proposition
\ref{prop:mergetwo2D} that any two of them can be merged in such a
way that the obtained graph is persistent and that the total
number of missing DOFs remains unchanged. Any pair of those
meta-vertices would indeed contain at least the required number of
DOFs. Doing this recursively, it is possible to merge all these
meta-vertices into {a single} persistent graph. In case there are
more than 3 missing DOFs, the total DOF number is by definition
smaller than $3\abs{N}+ 2\abs{S}-3$, which is the minimal number
of edges required to make the merged graph rigid. It follows then
from Theorem \ref{thm:leavedof2} that such meta-vertices cannot be
merged in a persistent graph {by addition of interconnecting
edges.} We have thus proved the following result:
\begin{prop}\label{prop:nec_and_suf-tomerge2D}
A collection of persistent meta-vertices $N\cup S$ {(with $N$ and
$S$ as defined in the beginning of this section)} can be merged
into a persistent graph if and only if the total number of missing
DOFs is no greater than 3, or equivalently if the total number of
local DOF in $N\cup S$ is at least $3\abs{N}+ 2\abs{S}-3$. At
least $3\abs{N}+ 2\abs{S}-3$ edges are needed to perform this
merging, {and merging} can always be done with exactly this number
of edges.\NDofthmprop
\end{prop}

As when merging rigid meta-vertices, we say that $G$ is an
\emph{{edge}-optimal persistent merging} if no single edge of
$E_M$ can be removed without losing persistence. Again, if all
meta-vertices are minimally persistent, then $G$ is an
{edge}-optimal persistent merging if and only if it is minimally
persistent.
\begin{thm}\label{thm:minper-minrig}
$G=\prt{\bigcup_{N,S} G_i}\cup E_M$ (with $N$ and $S$ as defined
at the beginning of this section and with all $G_i$ persistent) is
an {edge}-optimal persistent merging if and only if it is an
edge-optimal rigid merging and all edges of $E_M$ leave {vertices
with} local DOFs.
\end{thm}
\begin{proof}
Let $G$ be a persistent merging. If there is an edge that lies in
$E_M$ leaving a vertex with no local DOF, then it follows from
Proposition \ref{prop:d+} that the graph obtained by removing this
edge would also be persistent, and thus that $G$ is
not an edge-optimal persistent merging.\\
Now if $G$ is a persistent merging for which all edges of $E_M$
leave local DOFs but which is not an edge-optimal rigid merging,
then by removing one edge of $E_M$ it is possible to obtain a
rigid graph which by Proposition \ref{prop:necsufper_when_em<2} is
also persistent, so that $G$ is not an edge-optimal persistent
merging.\\
There remains to prove that an edge-optimal rigid merging $G$
where all edges of $E_M$ leave local DOFs is an edge-optimal
persistent merging. Since such $G$ is rigid, it follows from
Proposition \ref{prop:necsufper_when_em<2} that it is also
persistent. Moreover, since it is an edge-optimal rigid merging,
removing any edge of $E_M$ destroys rigidity and therefore
persistence.\\
\end{proof}

Tying Theorem \ref{thm:minper-minrig} with Theorem
\ref{thm:edge-optimal_rig2d} leads to the following more explicit
characterization of edge-optimal persistent merging.
\begin{thm} {$G=\prt{\bigcup_{N,S} G_i}\cup
E_M$ (with $N$ and $S$ as defined at the beginning of this section
and with all $G_i$ persistent)} containing at least two vertices
is an {edge}-optimal persistent merging in $\Re^2$ if and only if
{the following conditions all hold:}\\
(i) $\abs{E_M}= 3\abs{N}+2\abs{S}-3$.\\
(ii) For all {non-empty} $E_M''\subseteq E_M'$, there holds \\
\e\e $\abs{E_M''}\leq 3\abs{I(E_M'')}+2\abs{J(E_M'')}-3$\\
{with $I(E_M'')$ and $J(E_M')$ as defined in Theorem
\ref{thm:metalaman2D}}\\
(iii) All edges of $E_M$ leave {vertices with} local
DOFs.\NDofthmprop
\end{thm}
Notice that an efficient way to {obtain such a merging} is
provided in the discussion {immediately preceding} Proposition
\ref{prop:nec_and_suf-tomerge2D}.

\section{Rigidity and Persistence of 3D {Meta-Formations}} \label{sec:merge3D}

\subsection{Rigidity}\label{sec:3DRig}

We now consider a set $N$ of disjoint rigid (in $\Re^3$) graphs
$G_1,\dots,G_{\abs{N}}$ having at least three vertices each, a set
$D$ of graphs containing two (connected) vertices
$G_{\abs{N}+1},\dots,G_{\abs{N}+\abs{D}}$, and a set $S$ of
single-vertex graphs
$G_{\abs{N}+\abs{D}+1},\dots,G_{\abs{N}+\abs{D}+\abs{S}}$. As in
Section \ref{sec:merge2D}, these graphs are called meta-vertices,
and we define the merged graph $G$ by taking the union of all the
meta-vertices, and of some additional edges $E_M$ {each of
which has end-points belonging to different meta-vertices.}\\

\begin{table}\centering
\begin{tabular}{c|rrrrrr}
$\abs{V_a}$ &
\phantom{$\geq$}1&\phantom{$\geq$}1&\phantom{$\geq$}1&
\phantom{$\geq$}2 &\phantom{$\geq$}2&$\geq$ 3
\\\hline
$\abs{V_b}$& 1&2&$\geq$ 3 &2&$\geq$ 3 &$ \geq$ 3
\\\hline $\min \abs{E_M}$& 1 &2 &3&4&5&6
\end{tabular}
\caption{{Minimal number of edges required to merge two rigid
graphs $G_a$ and $G_b$ into a single rigid graph in $\Re^3$.}}
\label{tab:needed_edges}
\end{table}

The merging of two rigid meta-vertices, {each containing more than
two vertices,} is treated in \cite{YuFidanAnderson:2006_ACC}: {At
least six edges are needed, and they must be incident to at least
three vertices of each meta-vertex (which is necessary for
3-connectivity). But these conditions are only necessary, as they
do not imply 3-connectivity. For example, the so-called
\quotes{double-banana} graph in Fig. \ref{fig:rigidity}(c) can be
obtained by merging two distinct rigid tetrahedral meta-vertices
{(1,3,4,5) and (2,5,7,8)} using a total of six edges incident to
four vertices of each meta-vertex. However, it is always possible
to achieve a rigid merging using exactly six edges incident to
exactly three vertices of each meta-vertex, with no single vertex
having more than three incident edges out of the six.} {With a
minor modification, the merging result above} holds in the cases
where at least one meta-vertex has less than 3 vertices: {The
required number of edges is different, as summarized in Table
\ref{tab:needed_edges}} {where $\min \abs{E_M}$ represents the
minimal number of edges required to merge the meta-vertices
$G_a(V_a,E_a)$ and $G_b(V_b,E_b)$ into a rigid graph}. {Also,} if
a meta-vertex has less than 3 vertices, all of them should be
incident to edges of $E_M$, otherwise at least 3 of them should
be. When merging several meta-vertices, there is no available
necessary and sufficient condition for the rigidity of $G$.
Determining whether a merged graph is rigid in $\Re^3$ is indeed a
more general problem than determining whether a given graph is
rigid (for which it suffices to take $N=D = \varnothing$) {and
there is no} {known} {set of combinatorial necessary and
sufficient conditions for this}. We can however prove that the
rigidity of the merged graph $G$ only depends on $E_M$, on the
vertices to which nodes of $E_M$ are incident and on the belonging
of the $G_i$ to $N$, $D$ or $S$.

{\begin{prop}\label{prop:dependonlyEM3D} Let
$G=\prt{\bigcup_{N,D,S} G_i}\cup E_M$ with $N,D,S$ as defined at
the beginning of this section. Suppose that a meta-vertex $G_i$ is
replaced by a meta vertex $G_i'$ with the same set of vertices
incident to $E_M$, with the same set membership, $N$, $S$ or $D$,
as $G_i$, but otherwise with different internal structure. Let G'
be the graph so obtained. Then $G'$ is rigid if and only if $G$ is
rigid.
\end{prop}}

\begin{proof}
This could be proved using algebraic arguments based on the
rigidity matrix, but we prefer the following more intuitive
argument.\\

{The result is trivial for meta-vertices of $D$ and $S$ as they
are entirely determined by their belonging to these classes; we
assume therefore that $G_i\in N$. We also assume that the set
$V_i(E_M)$ of vertices of $G_i$ (and $G_i'$) that are incident on
edges of $E_M$ contains at least three vertices. In case this
assumption is not verified, both $G$ and $G'$ fail to be
3-connected and therefore rigid (by Theorem \ref{thm:Laman3D}), so
that the result is also trivial. We then prove that the
non-rigidity of $G$ implies the non-rigidity of $G'$. Since the
roles of $G$ and $G'$
can be exchanged, this is sufficient to prove the theorem.\\

Suppose that $G$ is not rigid, and give positions in $\Re^3$ to
its vertices. Then there is a smooth motion $M$ (satisfying the
distance constraints corresponding to edges in $G$) of the
vertices of $G$ apart from pure translation or rotation. Because
$G_i$ is rigid, the restriction of $M$ to the vertices of $G_i$ is
a rigid motion, that is a translation and/or rotation, which we
call $T$. Therefore, the restriction of $M$ to $(G\setminus
G_i)\cup V_i(E_M)$ is not a rigid motion. Otherwise all distances
would be preserved by $M$ apart from some distances between
vertices of $G_i\setminus V_i(E_M)$ and vertices of $G\setminus
G_i$. We would then have two vertices whose relative distance is
not preserved while their relative distance with respect to all
the three or more vertices of $V_i(E_M)$ are preserved, which is
impossible. We call $M^*$ this restriction to $(G\setminus
G_i)\cup V_i(E_M)$. Let now $M'$ be a smooth motion of the
vertices of $G'$, which for the vertices of $G_i'$ is the
translation and/or rotation $T$, and for the vertices of
$(G'\setminus G_i')\cup V_i(E_M)$ is the motion $M^*$ (observe
that that the two motions are identical on $V(E_M)$ which is the
intersection of the two sets on which $M'$ is defined). Since
$M^*$ is a non-rigid motion (not preserving all distances), so is
$M'$. Therefore, we just need to prove that $M'$ satisfies all
distance constraints on vertices connected by edges in $G'$ to
prove the non-rigidity of $G'$. Consider a pair of vertices. If
they both belong to $(G'\setminus G_i')\cup V_i(E_M)$, their
constraint in $G'$ is the same as in $G$, and their motion is
defined by $M^*$ which satisfies all distance constraints. If they
do not both belong to $(G'\setminus G_i')\cup V_i(E_M)$, then due
to the structure of the graph they necessarily both belong to
$G_i'$, and their motion is the rotation and/or translation which
by essence preserve all distances.}\\
\end{proof}

Moreover, we have the following necessary condition:
\begin{thm}\label{thm:metalaman3D} {Let $G_i$ for $i =
1,2,....,|N|+|D|+|S|$ be rigid meta-vertices, and suppose
$G=\prt{\bigcup_{N,D,S} G_i}\cup E_M$} {(with $N,D,S$ as defined
at the beginning of this section)} is rigid in
$\Re^3$ and contains at least three vertices. Then there exists $E_M'\subseteq E_M$ such that\\
(i) $\abs{E_M'}= 6\abs{N} + 5\abs{D}+ 3\abs{S} -6$\\
(ii) For all {non-empty} $E_M''\subseteq E_M'$, there holds\\
\e\e$\abs{E_M''}\leq 6\abs{I(E_M'')} + 5\abs{J(E_M'')}+
3\abs{K(E_M'')} -6$,\\
{where $I(E_M'')$ is the set of meta-vertices such that either
there are at least three vertices within the meta-vertex all
incident to edges of $E_M''$, or precisely two vertices within the
meta-vertex which are unconnected and both incident to edges of
$E_M''$. $J(E_M'')$ is the set of meta-vertices such that there
are precisely two vertices within the meta-vertex which are
connected and both incident to edges of $E_M''$; $K(E_M'')$ is the
set of meta vertices such that there is precisely one vertex
within the meta-vertex that is incident to one or several edges of
$E_M''$. Note that in each case, there can be an arbitrary number
of vertices in the meta-vertex which are not incident on any edge
of $E_M''$.}\\
Moreover, the graph $\prt{\bigcup_{N,D,S} G_i}\cup E_M'$ is rigid.
\NDofthmprop
\end{thm}
\begin{proof}
The proof is similar to the one of Theorem \ref{thm:metalaman2D}
(necessary part). For every $G_i$, let $G_i'$ be a minimally rigid
subgraph of $G_i$ on the same vertices, which therefore contains
$3\abs{V_i}-6$ edges if $G_i\in N$, one edge if $G_i \in D$ and no
edge if $G_i\subseteq S$. As mentioned in its proof, Lemma
\ref{lem:lem2metalaman} can also be applied in a three-dimensional
space. So if $G$ is rigid, there is a minimally rigid subgraph
$G'(V,E')\subseteq G$ containing all $G_i'$. Let $E_M'=E'\cap
E_M$; we shall prove that $E_M'$ satisfies the condition of this
theorem. Since $G'$ is minimally rigid, there holds
$\abs{E'}=3\abs{V}-6$. Moreover, we have $\abs{E'}= \abs{E_M'} +
\sum_{G_i\in N} \abs{E_i'}+\abs{D}$, and $\abs{V}  = \sum_{G_i\in
N} \abs{V_i'}+2\abs{D} + \abs{S}$, so that
\begin{equation*}
\begin{array}{ll}
\abs{E_M'} &= 3\abs{V}-6 - \sum_{G_i\in N} (3\abs{V_i}-6) -
\abs{D} \\& = 6\abs{N} + 5\abs{D} + 3\abs{S}-6.
\end{array}
\end{equation*}
$E_M'$ contains thus the predicted number of edges. We suppose now
that there is a set $E_M''$ such that $\abs{E_M''}>
6\abs{I(E_M'')}+5\abs{J(E_M'')}+3\abs{K(E_M'')}-6$ and {show that
this contradicts the minimal rigidity of $G'$}. Let us then build
$E''$ by taking the union of $E_M''$ and all $E_i'$ for which
$i\in I(E_M'')$, and the edge connecting the two vertices incident
to $E_M''$ in all meta-vertices in $J(E_M'')$. There holds
$V(E'')= \abs{K(E_M'')} + 2\abs{J(E_M'')}+ \sum_{G_i\in
I(E_M'')}\abs{V_i}$. Therefore, we have
\begin{equation*}
\begin{array}{lll}
\abs{E''} &=& \abs{E_M''} + \sum_{G_i\in I(E_M'')}\abs{E_i'}  +
\abs{J(E_M'')}\\& >&
6\abs{I(E_M'')}+5\abs{J(E_M'')}+3\abs{K(E_M'')}-6
\\&&+ \sum_{G_i\in I(E_M'')}\prt{3\abs{V_i}-6} + \abs{J(E_M'')}
\\&=& 3\abs{V(E'')} -6.
\end{array}
\end{equation*}
{This however contradicts the minimal rigidity of $G'$ as
$E''\subseteq E'$. Finally, since $G' = \prt{\bigcup_{N,D,S}
G'_i}\cup E_M'$ is rigid, it follows from several applications of
Proposition \ref{prop:dependonlyEM3D} that $\prt{\bigcup_{N,D,S}
G_i}\cup E_M'$ is also rigid.}\\
\end{proof}

Note that the rigidity of $\prt{\bigcup G_i}\cup E_M'$ is
explicitly mentioned here and not in Theorem
\ref{thm:metalaman2D}, because in a two-dimensional space it
follows directly from sufficiency of the counting conditions. But,
the counting conditions of Theorem \ref{thm:metalaman3D} are not
sufficient for rigidity, as the non-rigid graph of Fig.
\ref{fig:rigidity}(c) which can be obtained by merging two rigid
tetrahedral meta-vertices (1,3,4,5) and (2,6,7,8) would indeed
satisfy them. Nevertheless, one can deduce from Theorem
\ref{thm:metalaman3D} that $G$ is an edge-optimal rigid merging in
$\Re^3$ if and only if it is rigid and $\abs{E_M}=
6\abs{N}+5\abs{D}+3\abs{S}-6$, using $E_M'$ exactly in the same
way as in Theorem \ref{thm:edge-optimal_rig2d}.

\subsection{Persistence}\label{sec:3DPer}

{We consider now that all meta-vertices $G_i$ are persistent
graphs, and adapt the definitions of $N$, $D$ and $S$ in
consequence.} Theorem \ref{thm:leavedof2} can be generalized to
three dimensions, as {it follows from Proposition
\ref{prop:necsufper_when_em<2}, the proof of which can be
immediately extended to three dimensions.}

\begin{thm}\label{thm:leavedof3}
A collection of (structurally) persistent meta-vertices can be
merged into a (structurally) persistent graph if and only if it
can be merged into a (structurally) persistent graph by adding
edges leaving {vertices with one or more} local DOFs. In that
case, the merged graph is persistent if and only if it is
rigid.\NDofthmprop
\end{thm}
\begin{proof}
Suppose first that a collection of persistent meta-vertices can be
merged into a persistent graph in such a way that some edges do
not leave local DOFs. Then, it follows from Proposition
\ref{prop:d+} that these edges can be removed without destroying
the persistence of the merged graph, so that the same collections
of meta-vertices can be merged without having connecting edges
that do not leave local DOFs. In case the meta-vertices are
structurally persistent and are merged into a structurally
persistent graph, the result still holds as removing edges that do
not leave local DOFs never destroys structural persistence. The
reverse implication is
trivial.\\
The proof of the rest of the result is done exactly as in Theorem
\ref{prop:necsufper_when_em<2}, using Proposition
\ref{prop:dependonlyEM3D} instead of Theorem
\ref{thm:metalaman2D}.\\
\end{proof}

Merging two meta-vertices into a persistent graph is {however} a
more complicated problem in $\Re^3$ than in $\Re^2$. Consider
indeed a meta-vertex $G_a$ without any DOF, and a meta-vertex
$G_b$ which is not structurally persistent, i.e. which is
persistent and contains two vertices (leaders) having three DOFs.
The number of available DOFs is equal to the minimal number of
edges that should be added to obtain a rigid merged graph.
However, the only way to add six edges leaving local DOFs is to
add three edges leaving each leader of $G_b$ and arriving in
$G_a$, as {represented by the example in Fig.}
\ref{fig:merge3D_non_struc_per}(a). Only two vertices of $G_b$
would thus be incident to the added edges, which {prevents the
merged graph from being} rigid and therefore persistent {as it is
thus not 3-connected}. We have thus proved the following
condition:
\begin{prop}\label{prop:3D_imposs_merge}
If two persistent meta-vertices are such that one is not
structurally persistent and the other does not have any DOF, they
cannot be merged into a persistent graph {by addition of
interconnecting edges.}\NDofthmprop
\end{prop}
\begin{figure}
\centering
\begin{tabular}{cccc}
\includegraphics[scale = .35]{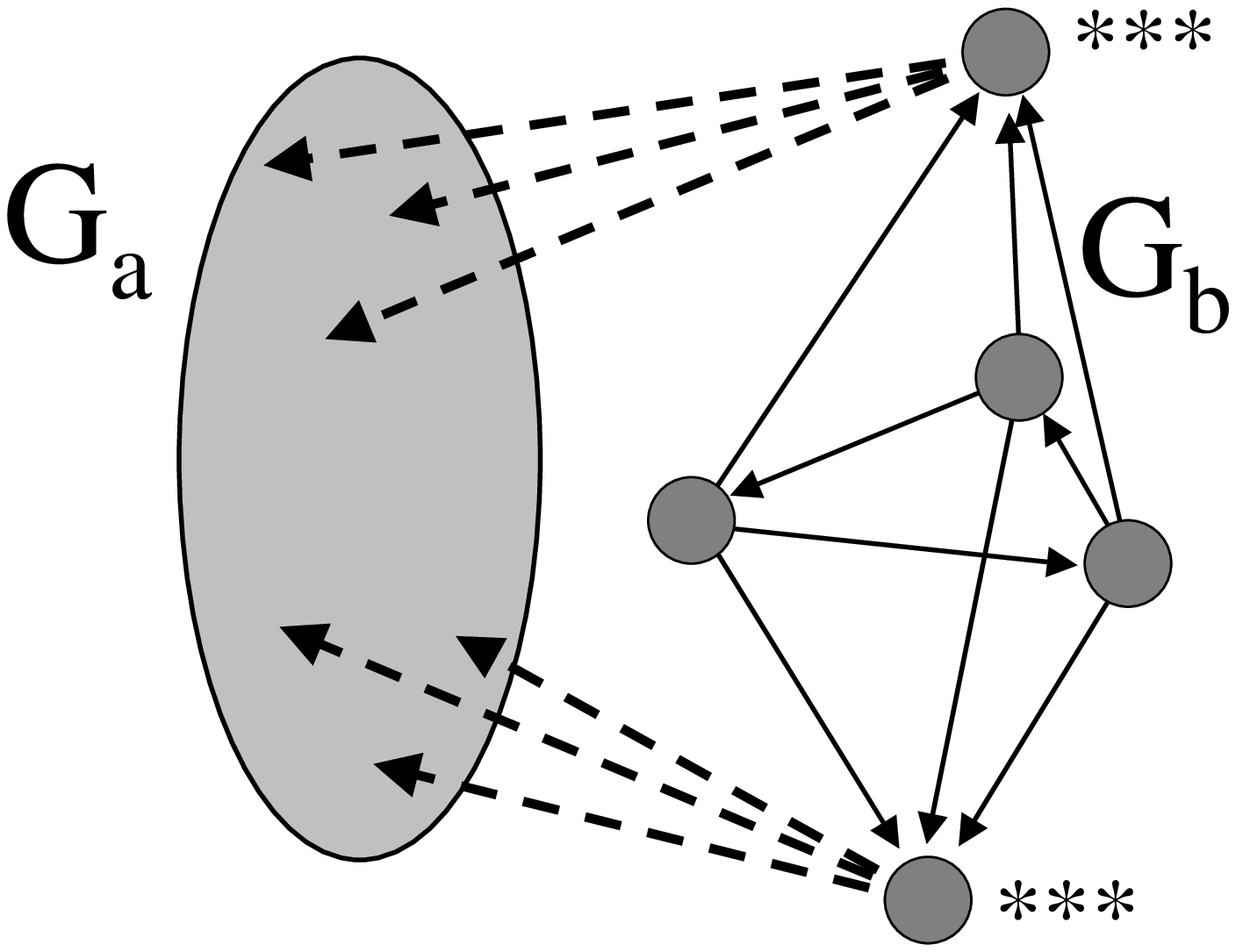}&&&
\includegraphics[scale = .35]{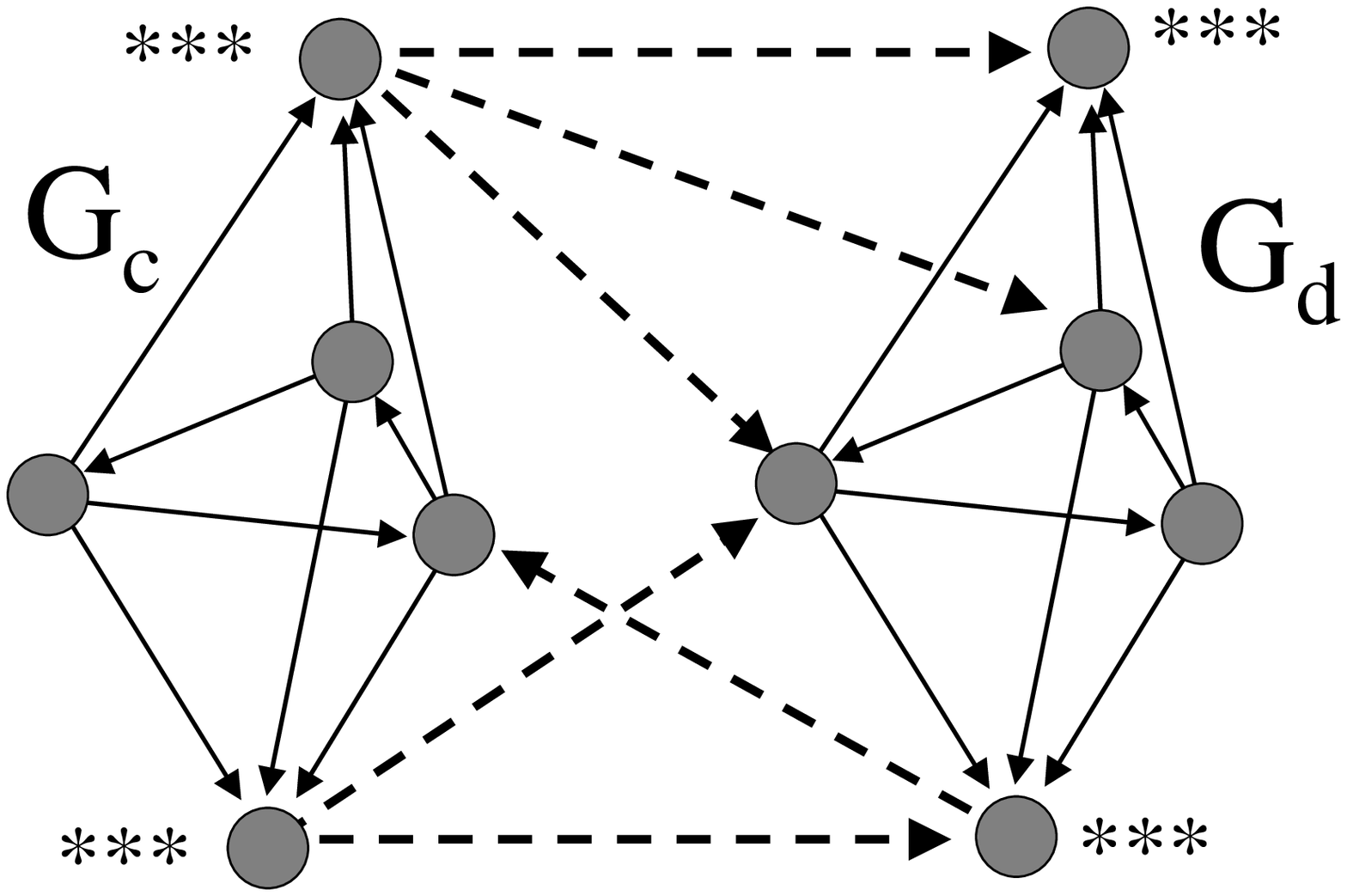}\\
(a) &&&(b)
\end{tabular}
\caption{Example of a {persistent but not} structurally persistent
meta-vertex $G_b$ which cannot be merged into a persistent or
rigid graph with the meta-vertex $G_a$, the latter {being
persistent but} having no DOF. (b) shows how two non-structurally
persistent meta-vertices can be merged into a structurally
persistent graph. The symbol \quotes{*} represents one DOF, and
the dashed edges are the edges of
$E_M$.}\label{fig:merge3D_non_struc_per}
\end{figure}

\begin{figure}
\centering
\includegraphics[scale = .35]{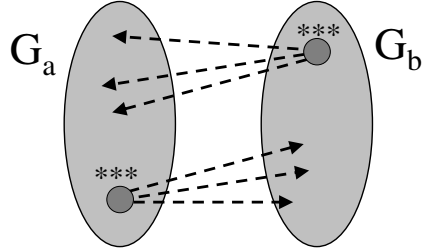}
\caption{$G_a$ and $G_b$ both have all their DOFs concentrated on
one leader. As a result they cannot be merged into a persistent
graph. The only way to add 6 edges leaving local DOFs is depicted
and does not lead to a rigid graph, because the overall graph is
not 3-connected. The symbol \quotes{*} represents one DOF, and the
dashed edges are the edges of
$E_M$.}\label{fig:merge3D_imposs2leader}
\end{figure}

Another problem appears when $G_a$ and $G_b$ each have one leader
(having three DOFs) and no other vertex has DOFs. Again, the
number of available DOFs is equal to the minimal number of edges
that should be added to obtain a rigid merged graph, but the only
way to add six edges leaving local DOFs does not lead to a rigid
graph. {One can indeed only add three edges leaving each leader as
shown in Fig. \ref{fig:merge3D_imposs2leader}. This results in a
graph that is not 3-connected and therefore not rigid by Theorem
\ref{thm:Laman3D}, as the removal of the two ex-leaders would
render the graph unconnected.} We have thus proved the following
condition:
\begin{prop}\label{prop:3D_imposs_merge_second}
If two persistent meta-vertices have each one leader (with 3 DOFs)
and no other DOF, they cannot be merged into a persistent graph
{by addition of interconnecting edges.}\NDofthmprop
\end{prop}

However, these are the only cases for which the argument {used in
establishing} Proposition \ref{prop:mergetwo2D} cannot be
generalized to establish an analogous property in $\Re^3$: {
\begin{prop}\label{prop:mergetwo3D}
Two persistent meta-vertices (each with three or more vertices)
can be merged into a persistent graph by addition of directed
connecting edges if and only if the sum of their DOFs is at least
6 and the DOFs are located on more than two vertices. At least six
edges are needed to perform this merging, {and merging can always
be done with exactly six edges and in such a way that the graph
obtained is structurally persistent and does not have all its DOFs
located on leaders.}
\end{prop}

\begin{proof}
Consider two meta-vertices each having more than 2 vertices. It
follows from Theorem \ref{thm:leavedof3} that they can be merged
into a persistent graph if and only if it is possible to add
directed edges leaving local DOFs in such a way that the obtained
graph is rigid.\\

{Suppose first that the total number of available DOFs is 6. If
all these DOFs are located on two leaders, the two graphs satisfy
the conditions of either Proposition \ref{prop:3D_imposs_merge} or
Proposition \ref{prop:3D_imposs_merge_second}, so that they cannot
be merged into a persistent graph. If the 6 DOFs are located on
more than 2 vertices, an exhaustive verification (see Appendix)
show that the two graphs can always be merged into a rigid graph
by adding 6 edges, each leaving a vertex with a local DOF, with at
least one DOF for each edge. Note that this exhaustive
verification is needed as no sufficient condition for rigidity of
a graph obtained by connecting two rigid graphs is known which is
sufficiently weak to be helpful for this proof.}\\

If the total number of DOFs is larger than 6, they are located on
at least 3 vertices, as a vertex has at most 3 DOFs. It is
therefore possible to select a subset of 6 DOFs located on at
least 3 vertices, and to apply the result obtained above for 6
DOFs.\\

{There remains to prove that the merging can always be done in
such a way that the obtained graph does not have all its DOFs
located on leaders, or in other words {the obtained graph has}
only vertices with 0 or 3 DOFs (This also implies that the graph
obtained is structurally persistent, as the only persistent graphs
that are not structurally persistent are those with two leaders
and therefore no other DOF). Such a situation, i.e. the obtained
graph has only vertices with 0 or 3 DOFs, could only happen if
this graph has exactly 3 or 6 DOFs, and thus if 9 or 12 DOFs are
initially available, as the merge is done by addition of 6 edges.
A simple way of avoiding having all remaining DOFs on leaders is
then to select the 6 DOFs that are going to be removed in the
merging process in such a way that a number of DOFs different from
3 and 6 is left in each of the initial graphs.
At least one vertex has then indeed one or two DOFs.}\\
\end{proof}
}

In case at least one of the two meta-vertices has less than 3
vertices, an exhaustive consideration of all possible cases {(see
Appendix)} shows that the result still holds, but with a different
required number of edges in $E_M$ and therefore of available DOFs:
these minimal numbers are both equal to $\min \abs{E_M}$ in Table
\ref{tab:needed_edges} (for the merging of a graph $G_a(V_a,E_a)$
with a graph $G_b(V_b,E_b)$). {Observe that as in the
2-dimensional case, the merge can be done in such a way that the
number of missing DOFs is preserved, the number of missing DOFs
being defined in the same way as in Section \ref{sec:2DPer}, with
maximal number of DOFs being 6, 5 and 3 for meta-vertices of
respectively $N$, $D$ and $S$.} It is worth noting that even if
one or both of the meta-vertices are not structurally persistent,
it is possible to obtain a structurally persistent merged graph,
as {represented} in Fig. \ref{fig:merge3D_non_struc_per}(b). This
has already been observed in
\cite{YuHendrickxFidanAndersonBlondel:2005} for the
case where one meta-vertex is a single vertex graph.\\

{Consider now a collection of meta-vertices such that the total
number of vertices is at least 3. Unless the collection consists
in two meta-vertices satisfying the hypotheses of Proposition
\ref{prop:3D_imposs_merge} or \ref{prop:3D_imposs_merge_second},
all the graphs that compose it can be merged into one large
persistent graph by addition of
edges.}\\
\begin{prop}\label{prop:nec_and_suf-tomerge3D}
A collection of {persistent meta-vertices $N\cup D\cup S$ (with
$N,D,S$ as defined in the beginning of this section)} containing
in total {at least three vertices and} that does not consist of
only two meta-vertices satisfying the condition of {Proposition
\ref{prop:3D_imposs_merge} or \ref{prop:3D_imposs_merge_second}}
can be merged into a persistent graph if and only if the total
number of missing DOFs is no greater than 6, or equivalently if
the total number of {local} DOFs in $N\cup D\cup S$ is at least
$6\abs{N}+5\abs{D} +3\abs{S}-6$. At least $6\abs{N}+5\abs{D}
+3\abs{S}-6$ edges are needed to perform this merging. {Merging}
can always be done with exactly this number of edges, and in such
a way that the merged graph is structurally
persistent.\NDofthmprop
\end{prop}

{
\begin{proof}
The proof is similar to the one of Proposition
\ref{prop:nec_and_suf-tomerge2D}. If a pair of meta-vertices can
be merged into a persistent graph, this merging can be done in
such a way that the number of missing DOFs is preserved, and by
adding only edges leaving vertices with local DOFs (with at most
one edge for each DOF). Doing this recursively, we eventually
obtain a single persistent graph that has the same number of
missing DOFs as the initial collection of graphs. The number of
added edges is then equal to the number of DOFs that have
disappeared during the
merging process, that is $6\abs{N}+5\abs{D} +3\abs{S}-6$.\\

There remains to prove that these mergings can actually be done,
and that the obtained graph is structurally persistent. By
Proposition \ref{prop:mergetwo3D} (and its extension to graphs
with 1 or 2 vertices), when their number of missing DOFs is
smaller than 6, two persistent graphs can always be merged into a
structurally persistent graph, unless either one of them is not
structurally persistent while the other has no DOF (case of
Proposition \ref{prop:3D_imposs_merge}), or both of them have one
leader and no other DOF (case of Proposition
\ref{prop:3D_imposs_merge_second}). In these two cases, the two
\quotes{problematic} meta-vertices have at least three vertices
each.\\

Suppose first that one meta-vertex has no DOF (and that the rest
of the meta-vertices collection does not consist in one single non
structurally persistent meta-vertex). Then since the total number
of missing DOF is 6, no other meta-vertex has a missing DOF, and
by hypothesis there are at least two other meta-vertices (or
possibly exactly one structurally persistent meta-vertex). It
follows then from successive applications of Proposition
\ref{prop:mergetwo3D} that they all can be merged into a
structurally persistent graph that still does not have any missing
DOF. This latter graph can then be merged with the graph that has
no DOF, and the graph obtained is
also structurally persistent.\\

Suppose now that two meta-vertices have exactly one leader and no
other DOF. It follows then from the hypotheses that there is at
least one other meta-vertex in the collection. And again, no other
meta-vertex has any missing DOF. Temporarily isolating one of the
meta-vertices with one leader and no other DOF, it follows again
from successive applications of Proposition \ref{prop:mergetwo3D}
that all other graphs can be merged into a persistent graph that
does not have all its DOFs located on one single leader, and this
graph can then be merged with the temporarily isolated graph into
a structurally persistent graph.\\
\end{proof}
}

As in the two-dimensional case, a merged graph is an
{edge}-optimal persistent merging if and only if it is an
{edge}-optimal rigid merging and all edges in $E_M$ (such as
defined in the beginning of this subsection) leave local DOFs. The
proof of this is an immediate generalization of Theorem
\ref{thm:minper-minrig}. However, due to the absence of necessary
and sufficient conditions {allowing a combinatorial checking of}
the rigidity of a graph or of a merged graph in $\Re^3$, the
result cannot be expressed in a purely combinatorial way. Since
the number of edges in $E_M$ in an edge-optimal rigid merging is
fixed, the above criterion can be re-expressed as
\begin{thm}
{$G=\prt{\bigcup_{N,D,S} G_i}\cup E_M$ (with $N,D,S$ as defined at
the beginning of this section and with all $G_i$ persistent)}
containing in total at least three vertices is an {edge}-optimal
persistent merging in
$\Re^3$ if and only if {the following conditions all hold:}\\
(i) $G$ is rigid. \\
(ii) All edges of $E_M$ leave local DOFs.\\
(iii) $\abs{E_M}= 6\abs{N}+5\abs{D}+3\abs{S}-6$.\\
\NDofthmprop
\end{thm}

{Again, an efficient way to obtain an {edge}-optimal persistent
merging from a collection of meta-vertices satisfying the
hypotheses of Proposition \ref{prop:nec_and_suf-tomerge3D} is to
first merge two of them and then to iterate, as in the discussion
of Propositions \ref{prop:nec_and_suf-tomerge2D} and
\ref{prop:nec_and_suf-tomerge3D}.}

\section{Conclusions}\label{sec:concl}

We have analyzed the conditions under which a formation resulting
from the merging of several persistent formations is itself
persistent. Necessary and sufficient conditions were found to
determine which collections of persistent formations could be
merged into a larger persistent formation. We first treated these
issues in $\Re^2$. Our analysis was then generalized to $\Re^3$
and to structural persistence, leading to somewhat less powerful
results. {This is especially the case for those which rely on the
sufficient character of Laman's conditions for rigidity in $\Re^2$
{(Theorem \ref{thm:Laman})}, no equivalent condition being known
in $\Re^3$.} Following this work, we plan to {develop} systematic
ways to build all possible optimally merged persistent formations,
similarly to what has been done for minimally persistent
formations \cite{HendrickxFidanYuAndersonBlondel:2006} and for
minimally rigid merged formations
\cite{YuFidanHendrickxAnderson:2006_CDC}. These references canvas
generalizations of the Henneberg sequence concept
\cite{Henneberg:11,TayWhiteley:85} for building all minimally
rigid graphs in two dimensions.

\appendix

In this appendix, we complete the proof of Proposition
\ref{prop:mergetwo3D} on the merging of persistent meta-vertices
in a 3-dimensional space, {and extend this proposition to cases
where one of the meta-vertices has less than 3 vertices}. We have
to prove that two persistent graphs (in a three-dimensional space)
$G_a$ and $G_b$ having in total 6 DOFs located on at least three
vertices can always be merged into a rigid graph by addition of
six edges leaving vertices with local DOFs, with at least one
local DOF for each added edge.\\

For this purpose, we use the following lemma, which summarizes
results obtained in \cite{YuFidanAnderson:2006_ACC}.
\begin{lem}\label{lem:2op}
Let $G_a$ and $G_b$ be two (initially distinct) rigid graphs each
with three or more vertices. Performing a sequence of three or
more operations selected among the two following types of
operations results in merging $G_a$ and $G_b$ into a rigid graph
by addition of 6 edges.\\
\emph{Operation (v):} Taking a vertex $i$ of $G_a$ not connected
yet to any vertex of $G_b$, and connecting it to $3-t$ vertices of
$G_b$, where $t$ is the number of operations already performed.\\
\emph{Operation (e):} Taking a vertex $i$ of $G_a$ not connected
yet to any vertex of $G_b$, and an edge $(k,j)$ with $k\in V_A$
and $j\in V_B$. Replacing the edge $(k,j)$ by $(i,j)$ and
connecting $i$ to $2-t$ other vertices in $G_b$, where $t$ is the
number of operations already performed.\\
\end{lem}

Without loss of generality, we suppose that $G_a$ has at least as
many DOFs as $G_b$. The partition of DOFs can thus be 6-0, 5-1,
4-2 or 3-3. In the sequel, we prove the result for each of these
particular cases, starting with $G_a$ having 6 DOFs.\\

It follows from Lemma \ref{lem:2op} that the merged graph
represented in Fig. \ref{fig:gagb-60-321}(a) is rigid. It can
indeed be obtained by three applications of the operation (v).
Moreover, one can see in Fig. \ref{fig:gagb-60-321}(b) that
directions can be given to the connecting edges in such a way that
the out-degree distribution (with respect to the connecting edges)
is $(3,2,1)$, that is one vertex of $G_a$ is left by three
connecting edges, one by two, and one by one. Suppose now that
$G_a$ is a persistent graph with 6 DOFs with a DOF allocation
$(3,2,1)$, that is a persistent graph having one vertex having 3
DOFs, one 2 DOFs, and one 1 DOF. Then it can be merged with $G_b$
into a rigid graph by adding 6 edges leaving vertices with local
DOF (with one DOF for each edge). It suffices indeed to take the
edges represented in Fig. \ref{fig:gagb-60-321}(b), identifying
each vertex with $\delta$ DOFs with a vertex left by $\delta$
connecting
edges.\\

\begin{figure}
\centering
\begin{tabular}{cc}
\includegraphics[scale= 0.5]{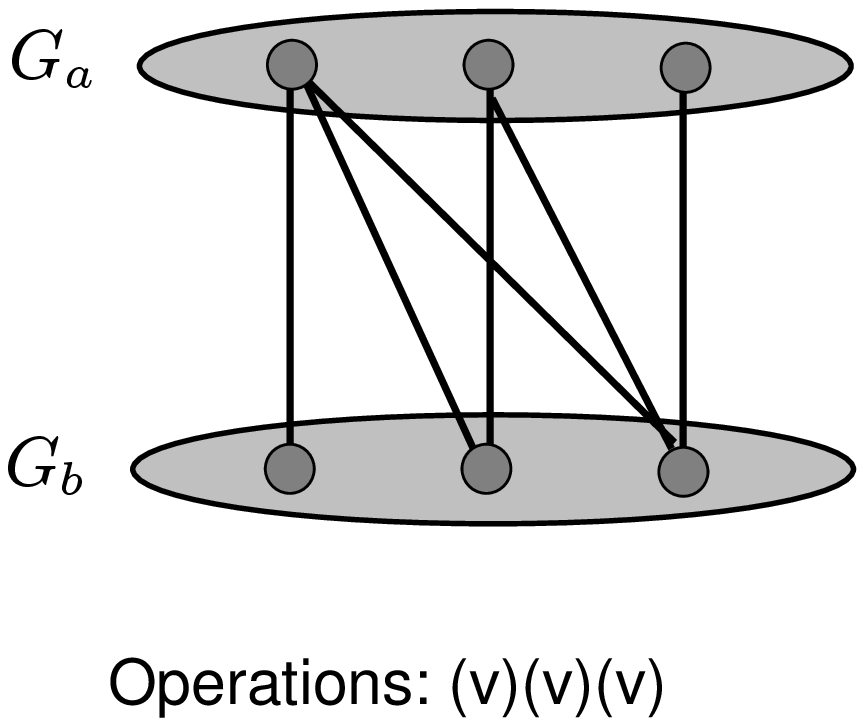}
&
\includegraphics[scale= 0.5]{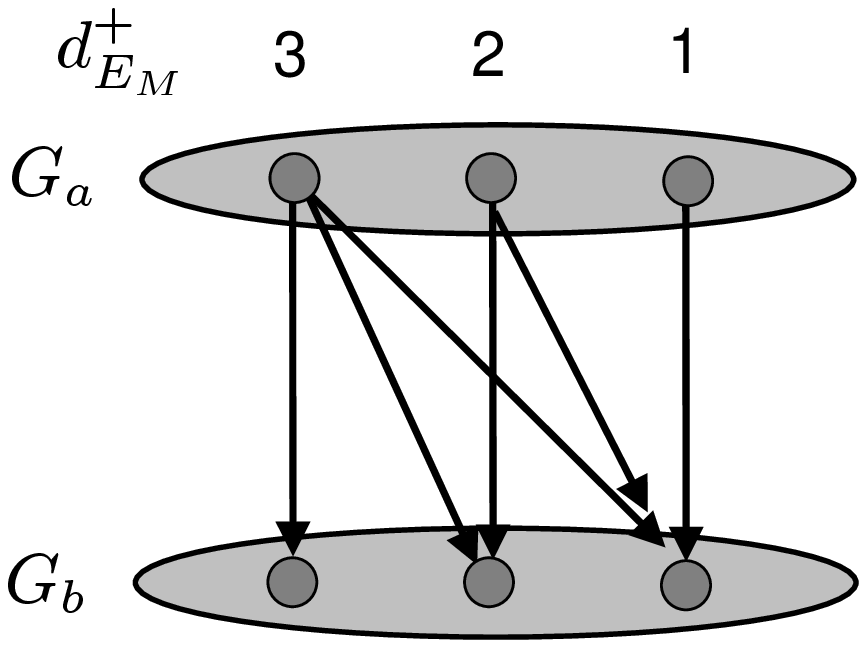}
\\
(a)&(b)
\end{tabular}
\caption{(a) represents a rigid merged graph obtained by
performing three operations (v). It is shown in (b) how directions
can be given to the edges in such a way that three vertices of
$G_a$ are left by respectively 3, 2 and 1 edges.}
\label{fig:gagb-60-321}
\end{figure}

We now treat a DOF allocation $(2,2,2)$. It follows again from
Lemma \ref{lem:2op} that the merged graph represented in Fig.
\ref{fig:gagb-60-222}(a) is rigid, as it can be obtained by two
applications of the operation (v) followed by one application of
operation (e). Moreover, Fig. \ref{fig:gagb-60-222}(b) shows that
directions can be assigned to the edges in such a way that the
out-degree distribution (with respect to the connecting edges) is
$(2,2,2)$. For the same reason as above, $G_a$ can thus be merged
with $G_b$ into a rigid graph by adding 6 edges leaving vertices
with local DOF (with one DOF for each edge) if its DOF
distribution is $(2,2,2)$.\\

\begin{figure}
\centering
\begin{tabular}{cc}
\includegraphics[scale= 0.5]{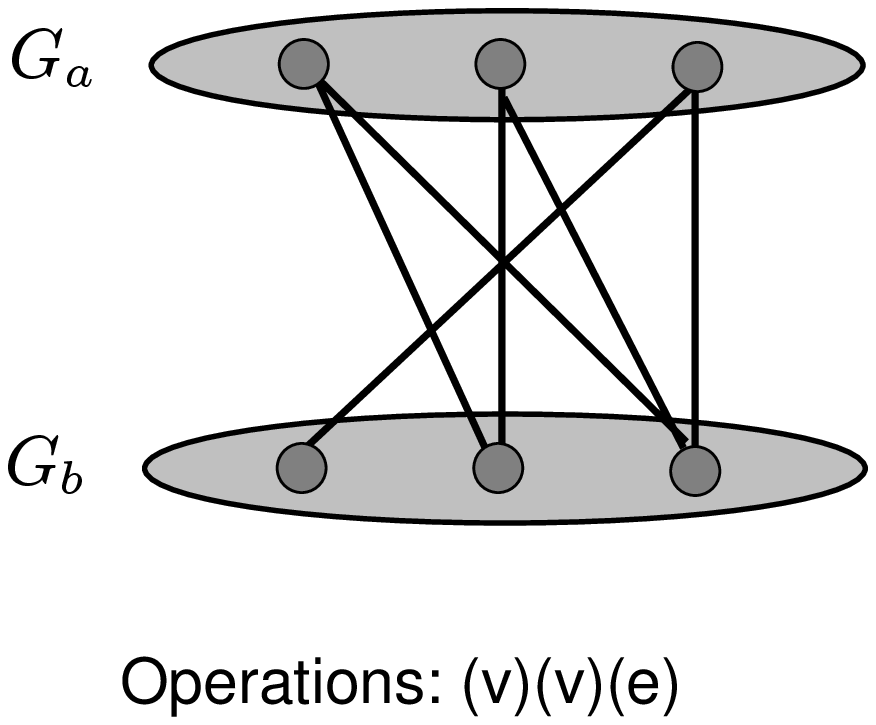}
&
\includegraphics[scale= 0.5]{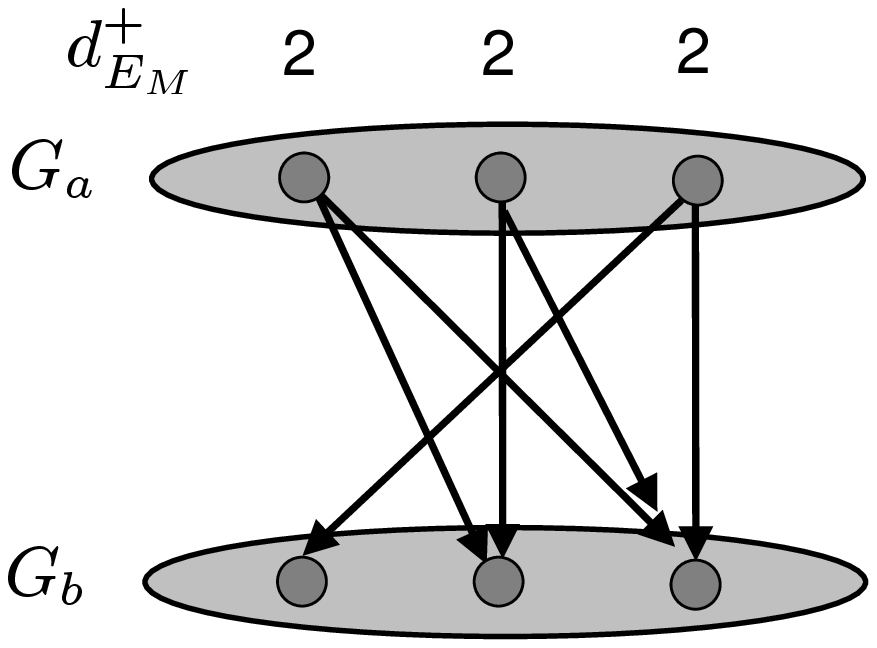}
\\
(a)&(b)
\end{tabular}
\caption{(a) represents a rigid merged graph obtained by
performing three operations (v). It is shown in (b) how directions
can be given to the edges in such a way that the three vertices of
$G_a$ are each left by 2 edges.} \label{fig:gagb-60-222}
\end{figure}

Next we show that such construction can be obtained in all other
cases, except those where the 6 DOFs are all located on two
vertices. When $G_a$ has 6 DOFs, the remaining possible DOF
distributions are $(3,1,1,1)$, $(2,2,1,1)$, $(2,1,1,1,1)$ and
$(1,1,1,1,1,1)$, the case $(3,3)$ {does not satisfy the}
hypotheses. The construction in these four cases are obtained by
performing the operation (e) of Lemma \ref{lem:2op} (up to three
times) on the constructions detailed above for $(3,2,1)$ and
$(2,2,2)$. They are represented in Fig. \ref{fig:gagb-60-others}\\

\begin{figure}
\centering
\begin{tabular}{cc}
\includegraphics[scale= 0.4]{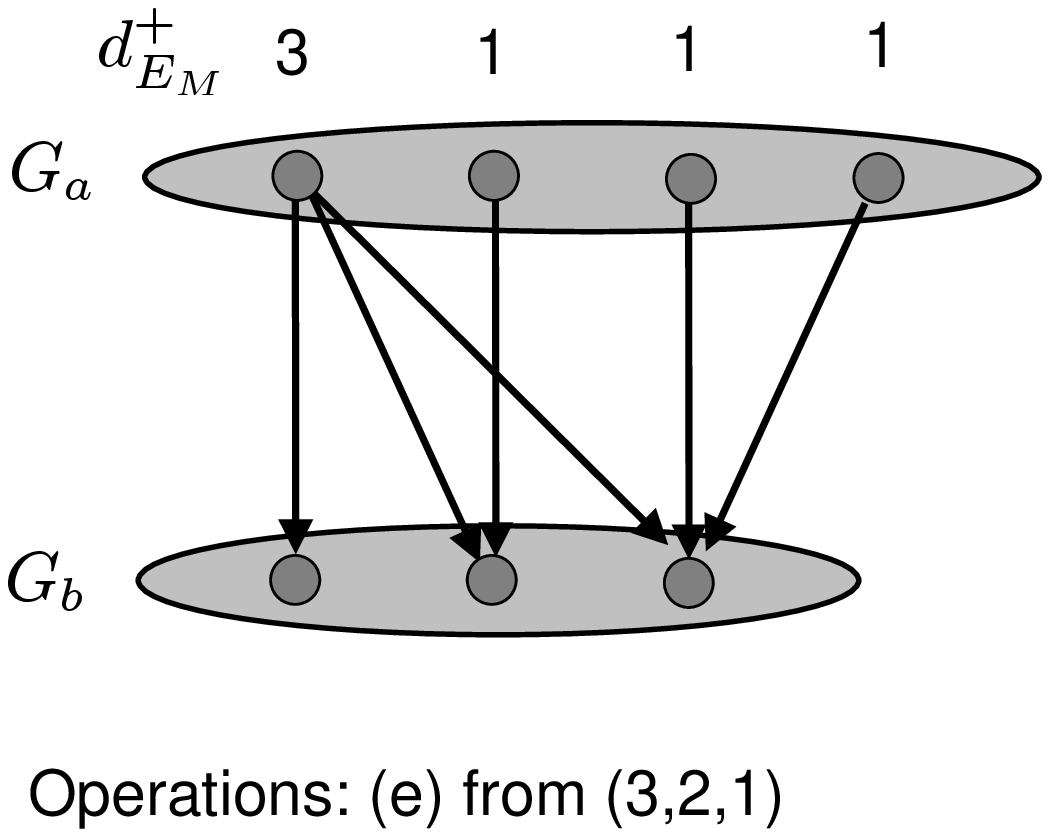}&
\includegraphics[scale= 0.4]{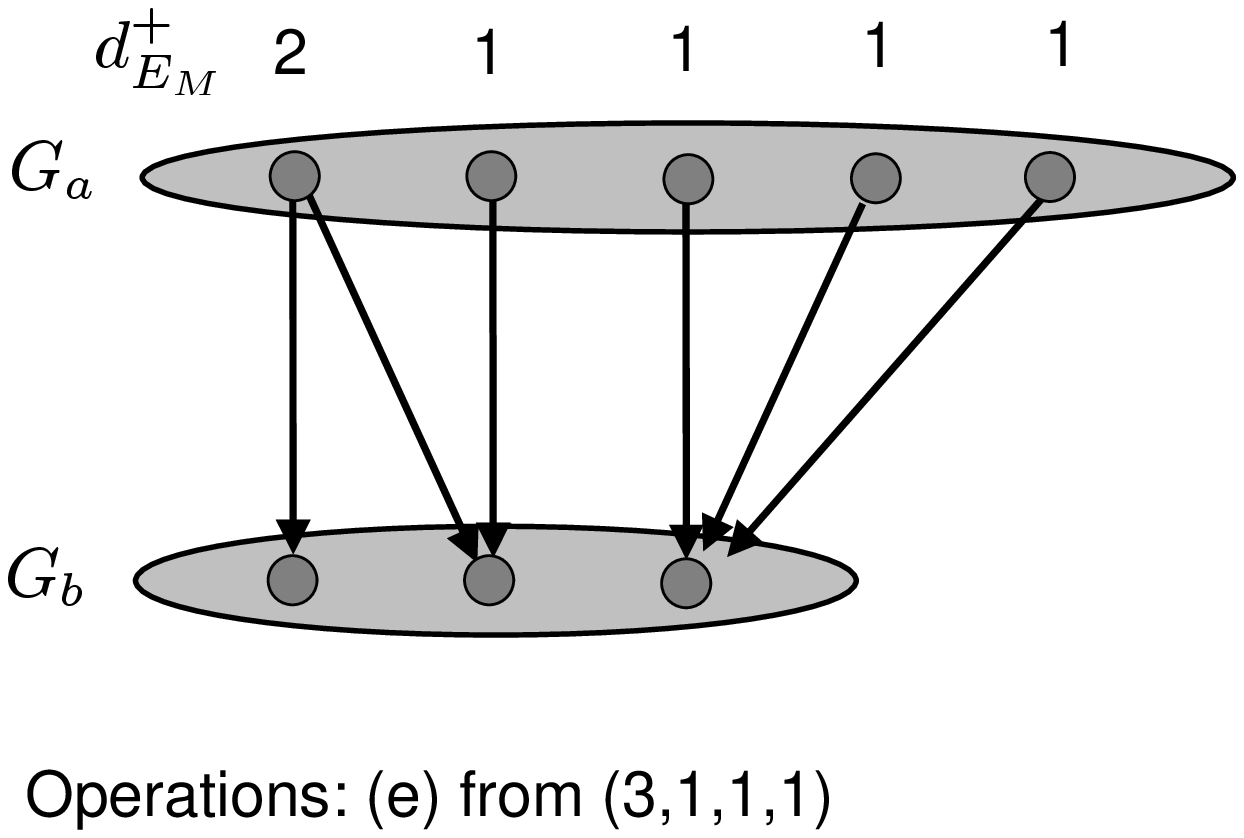}\\
\includegraphics[scale= 0.4]{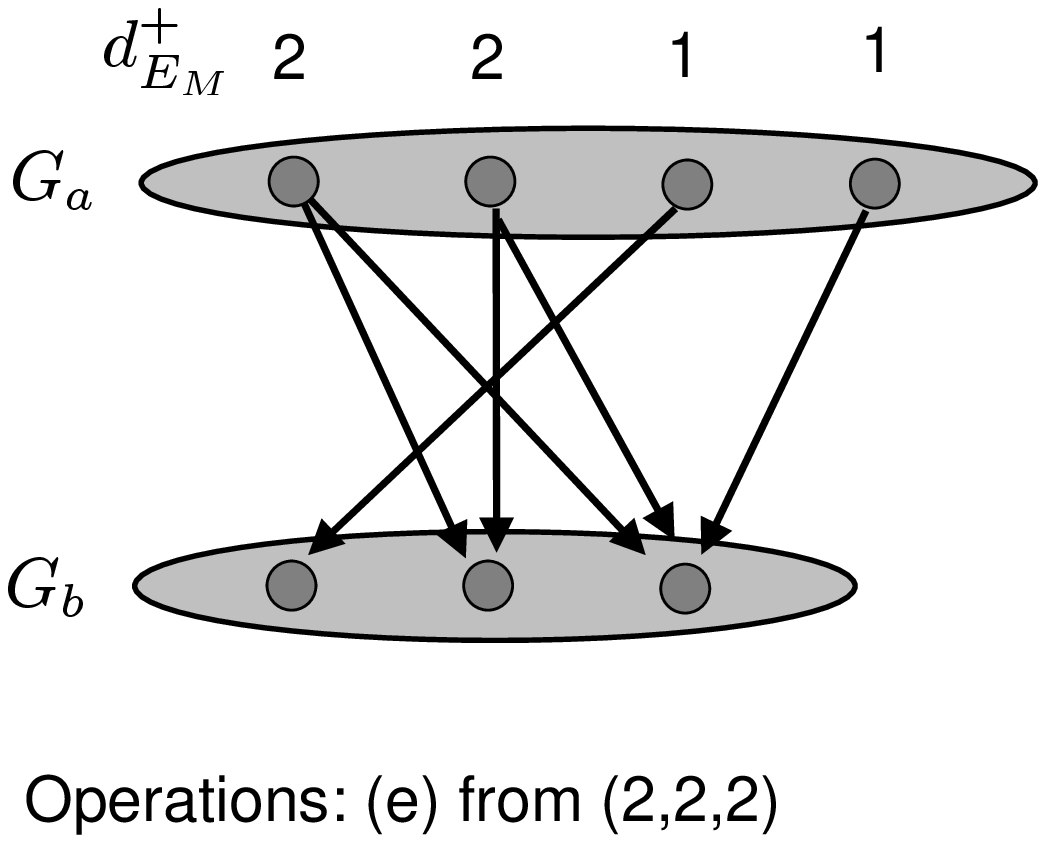}&
\includegraphics[scale= 0.4]{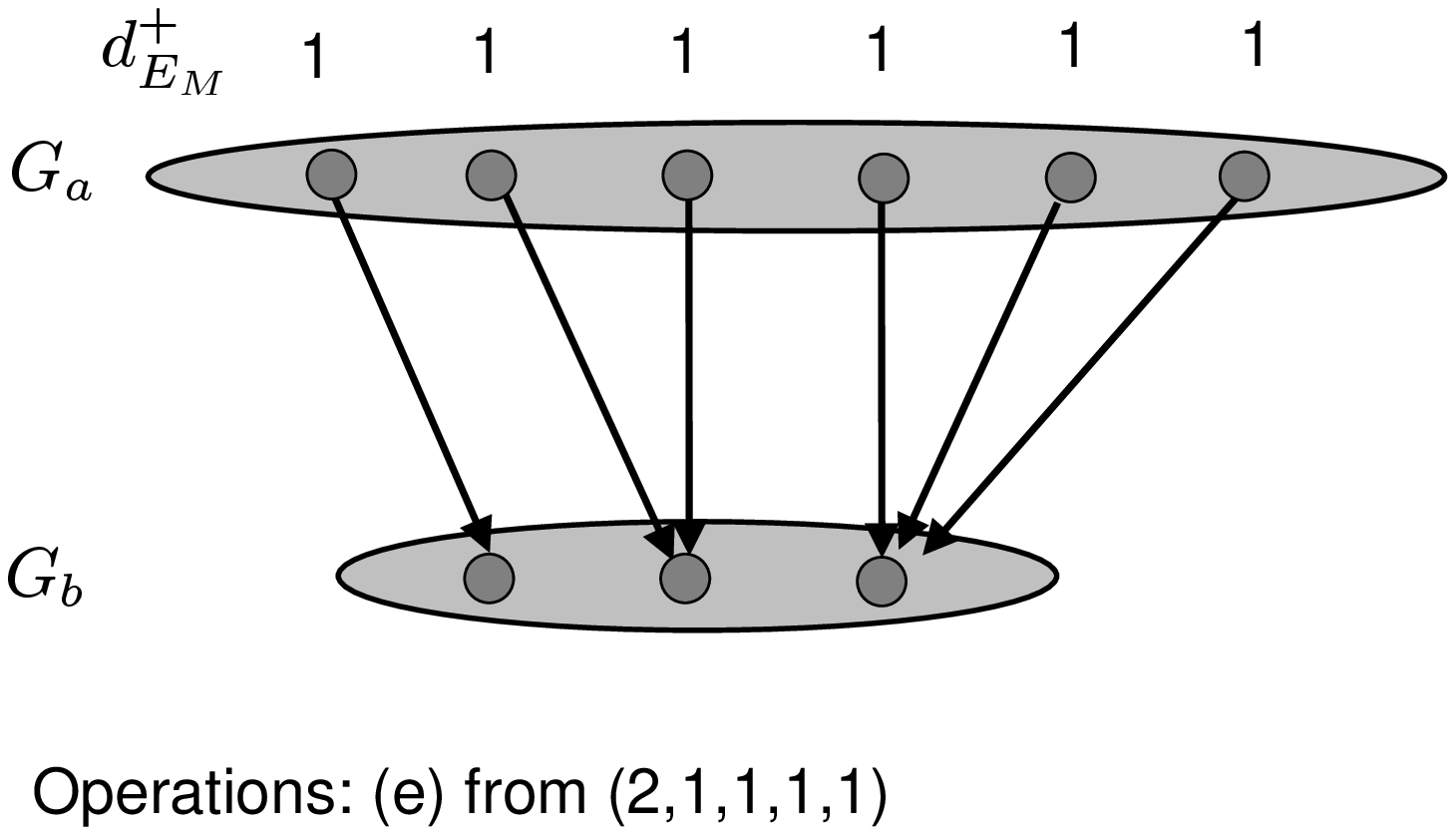}
\end{tabular}
\caption{Representations of how a rigid graph can be obtained by
merging two persistent graphs $G_a$ and $G_b$ where $G_b$ has no
DOF and where the DOF allocation of $G_a$ is $(3,1,1,1)$,
$(2,2,1,1)$, $(2,1,1,1,1)$ or $(1,1,1,1,1,1)$. The operations (v)
and (e) used to obtain the structure are also mentioned. }
\label{fig:gagb-60-others}
\end{figure}

If $G_a$ has 5 DOFs and $G_b$ one DOF, the required construction
can always be obtained from one of the construction for the case
where $G_a$ has 6 DOFs. {It suffices indeed to use one of the
constructions already provided by temporarily adding one vertex
with one DOF to the distribution of 5 DOFs in $G_a$, executing the
appropriate construction from the group above, and then reversing
the direction of the edge leaving a vertex with one DOF, as shown
in Fig. \ref{fig:gagb-51} for a DOF distribution
$(3,2)$.}\\

\begin{figure}
\centering
\includegraphics[scale= 0.5]{GAGB-60-321b.eps}
\includegraphics[scale= 0.5]{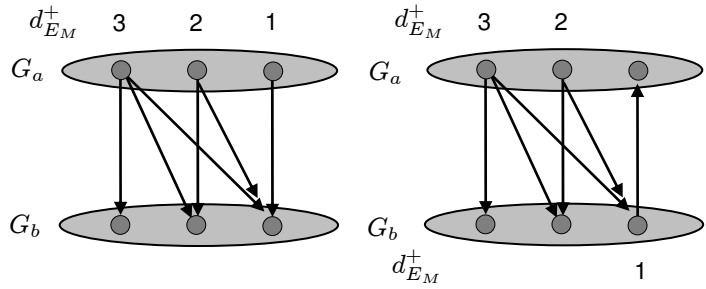}
\caption{Representation of how the construction for a DOF
partition 5-1 between $G_a$ and $G_b$ can be obtained from a
construction for a partition 6-0.} \label{fig:gagb-51}
\end{figure}

Suppose now that $G_a$ has 4 DOFs, and $G_b$ 2 DOFs. Then the
possible DOF distribution for $G_a$ are $(3,1)$, $(2,2)$,
$(2,1,1)$ and $(1,1,1,1)$. For $G_b$, they are $(2)$ and $(1,1)$.
The construction proving the result for these eight cases are
shown in Fig. \ref{fig:gagb-42}.\\

\begin{figure}
\centering
\includegraphics[scale= 0.4]{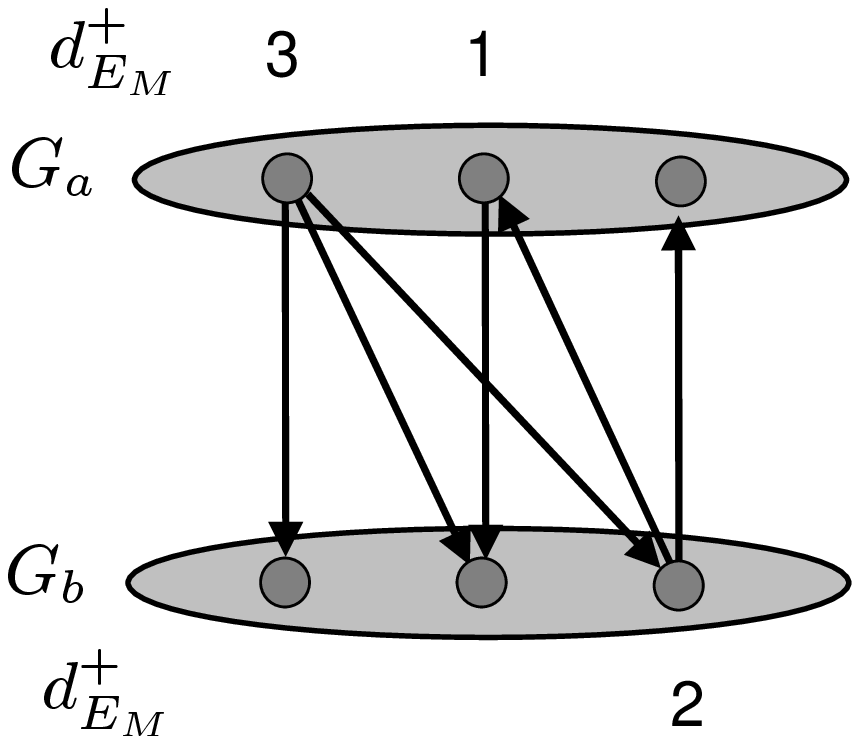}
\includegraphics[scale= 0.4]{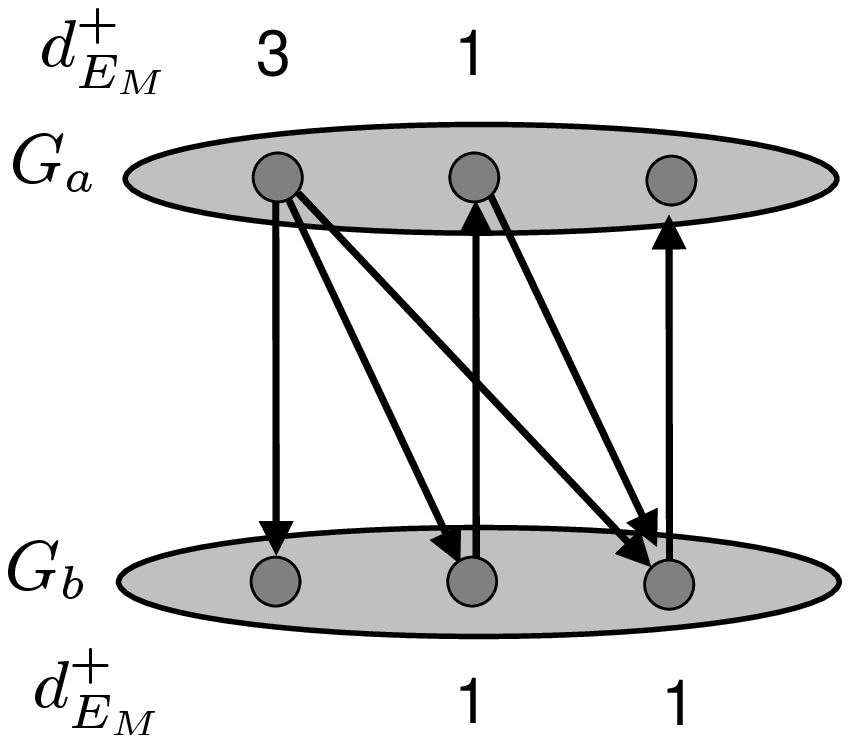}
\includegraphics[scale= 0.4]{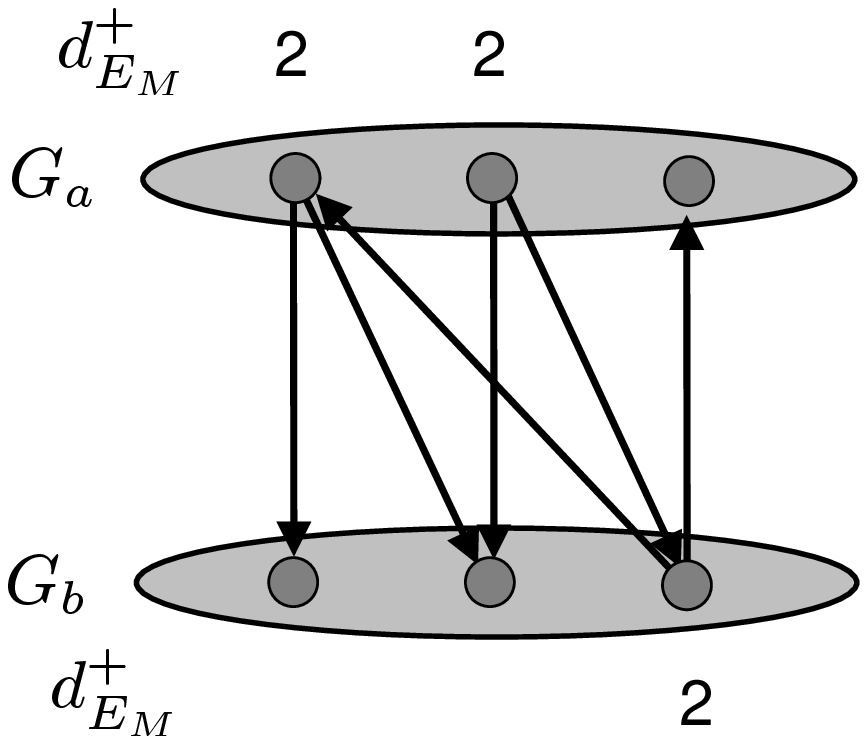}\\
\includegraphics[scale= 0.4]{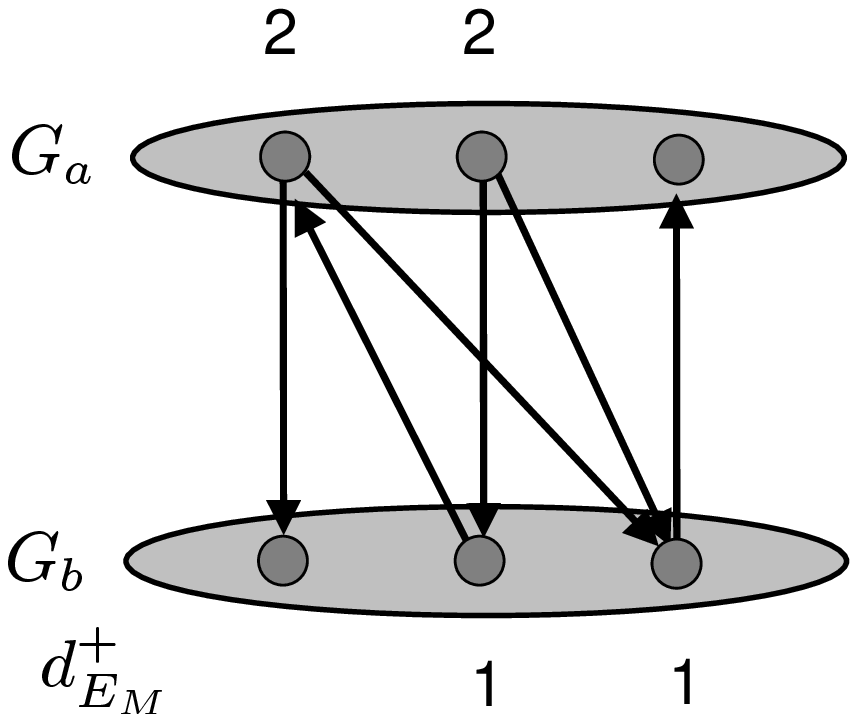}
\includegraphics[scale= 0.4]{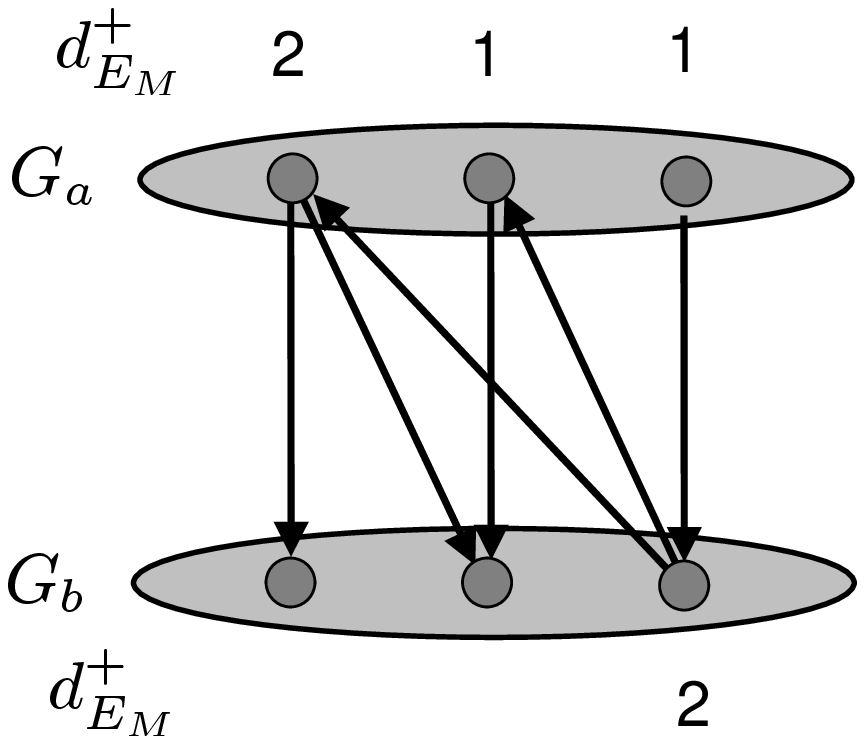}
\includegraphics[scale= 0.4]{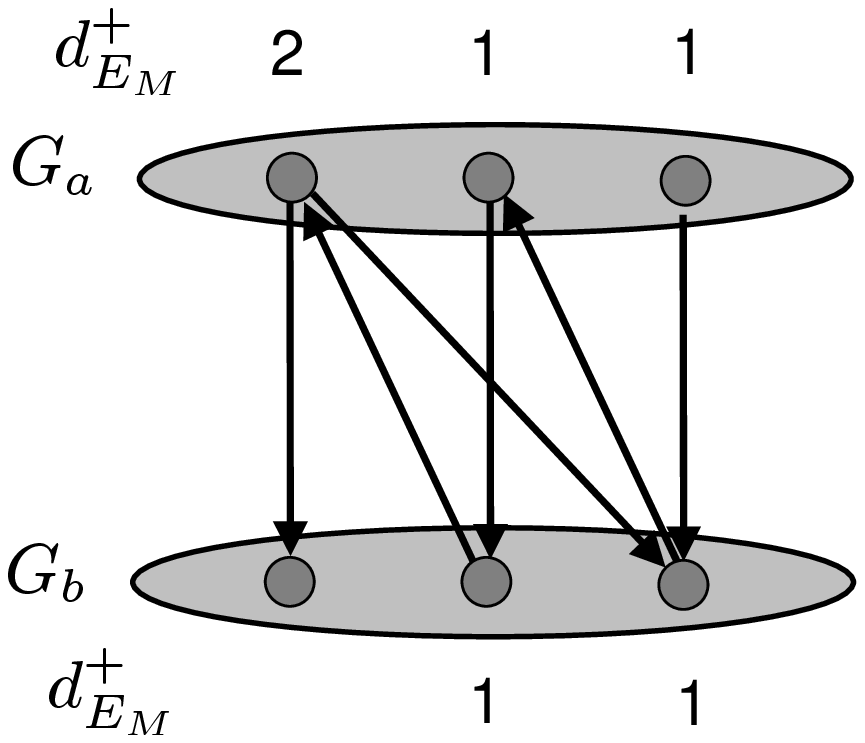}\\
\includegraphics[scale= 0.4]{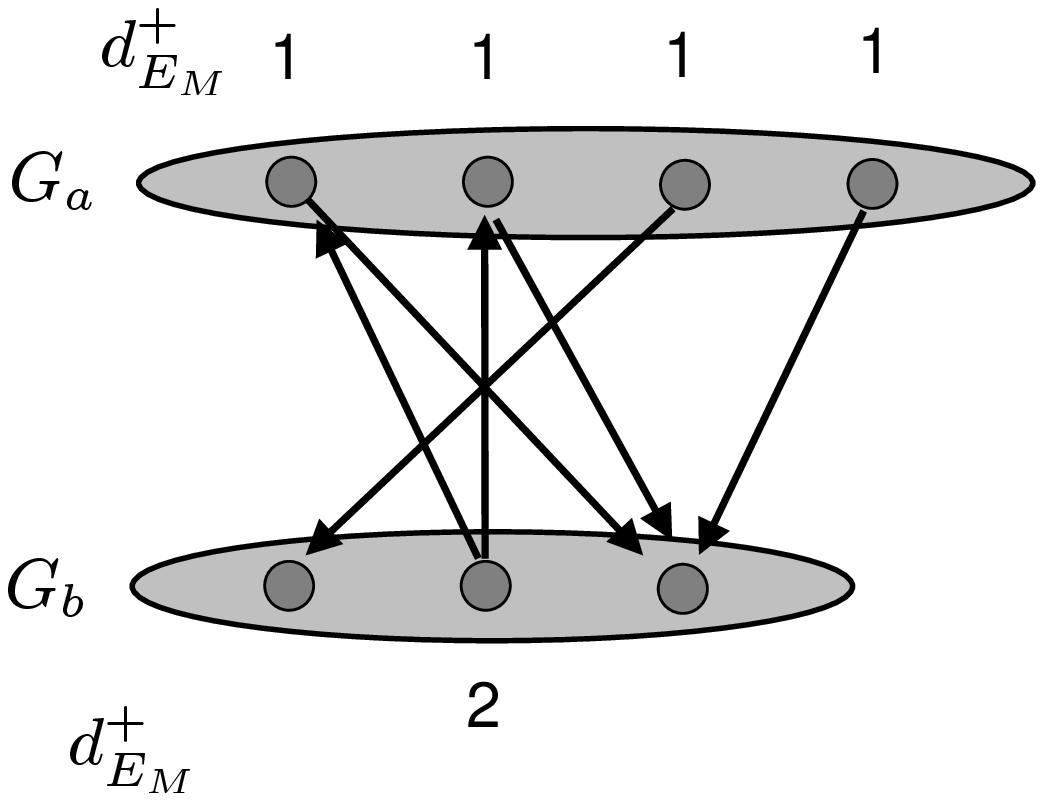}
\includegraphics[scale= 0.4]{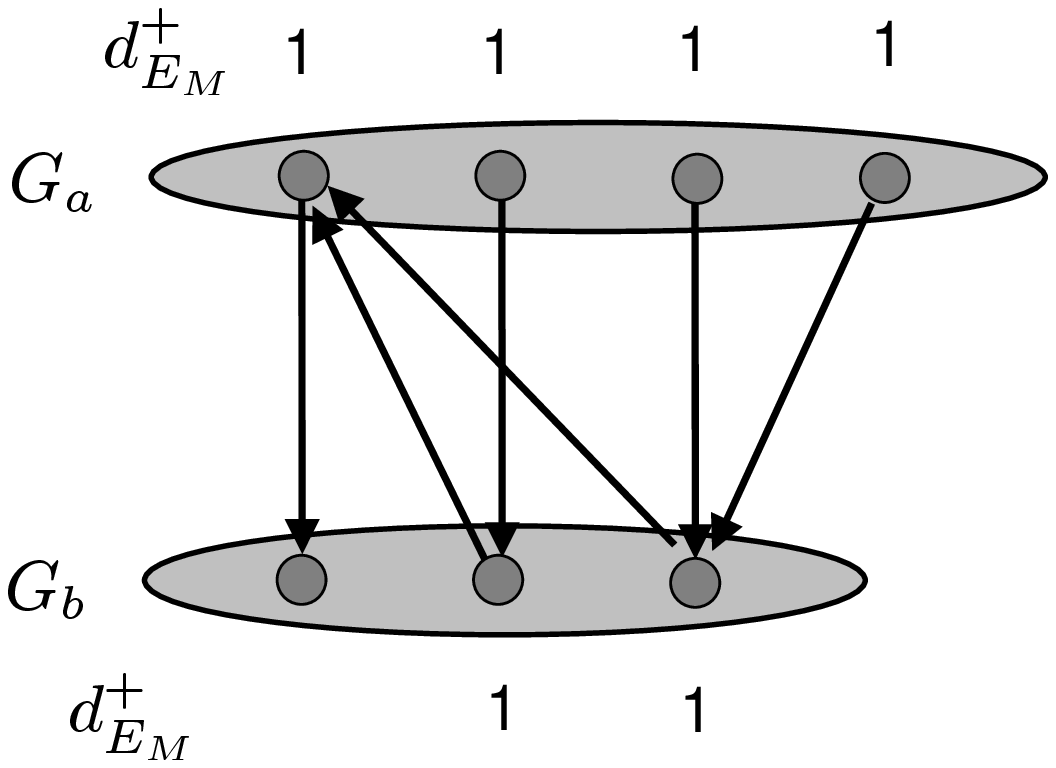}\\
\caption{Constructions for the eight possible DOF allocations when
$G_a$ has 4 DOFs and $G_b$ 2 DOFs. The graphs are all rigid are
they have the same undirected underlying graphs as construction in
Fig. \ref{fig:gagb-60-321}, \ref{fig:gagb-60-222} or
\ref{fig:gagb-60-others}.} \label{fig:gagb-42}
\end{figure}

Finally, if both graphs have 3 DOFs, the possible distribution for
each are $(3)$, $(2,1)$ and $(1,1,1)$. The case where they both
have a distribution $(3)$ does not satisfy the hypotheses of this
Proposition, and three other cases do not need to be treated for
symmetry reasons. The construction for the remaining 5 cases is
shown in Fig. \ref{fig:gagb-33}.\\

\begin{figure}
\centering
\includegraphics[scale= 0.4]{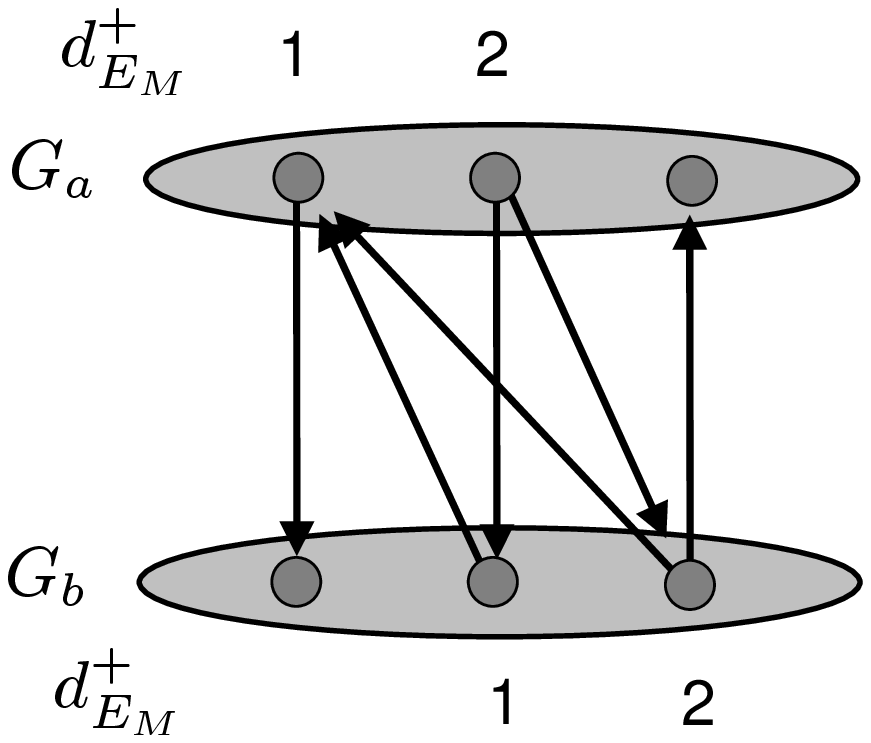}
\includegraphics[scale= 0.4]{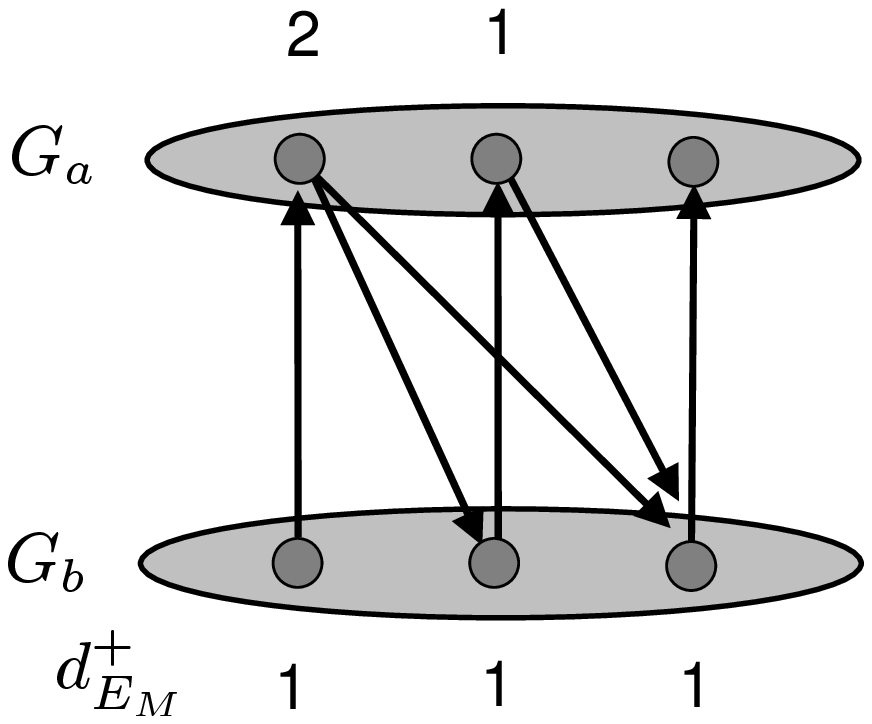}
\includegraphics[scale= 0.4]{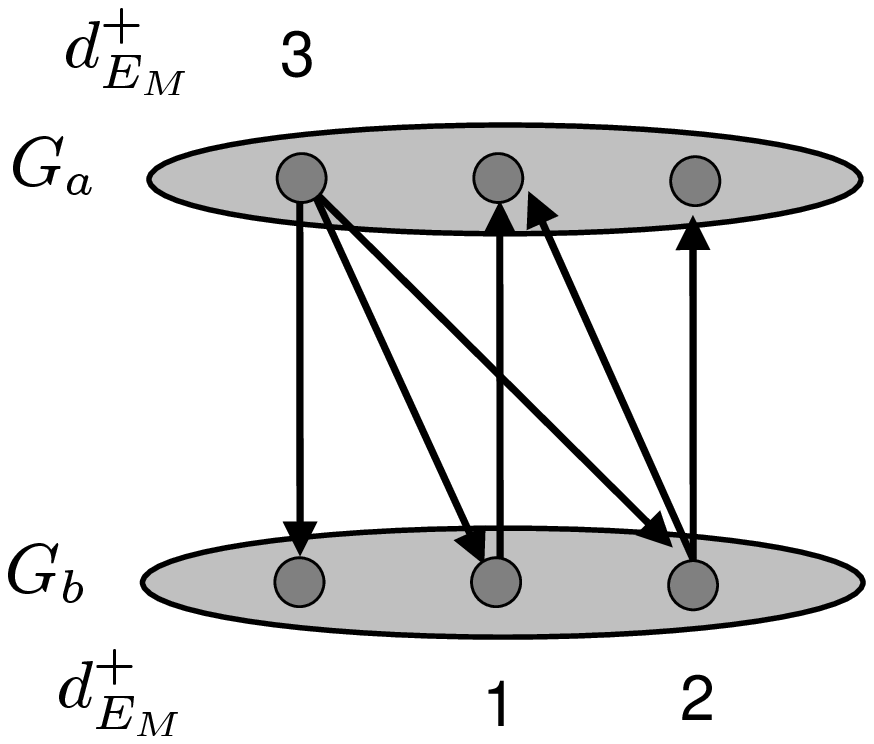}\\
\includegraphics[scale= 0.4]{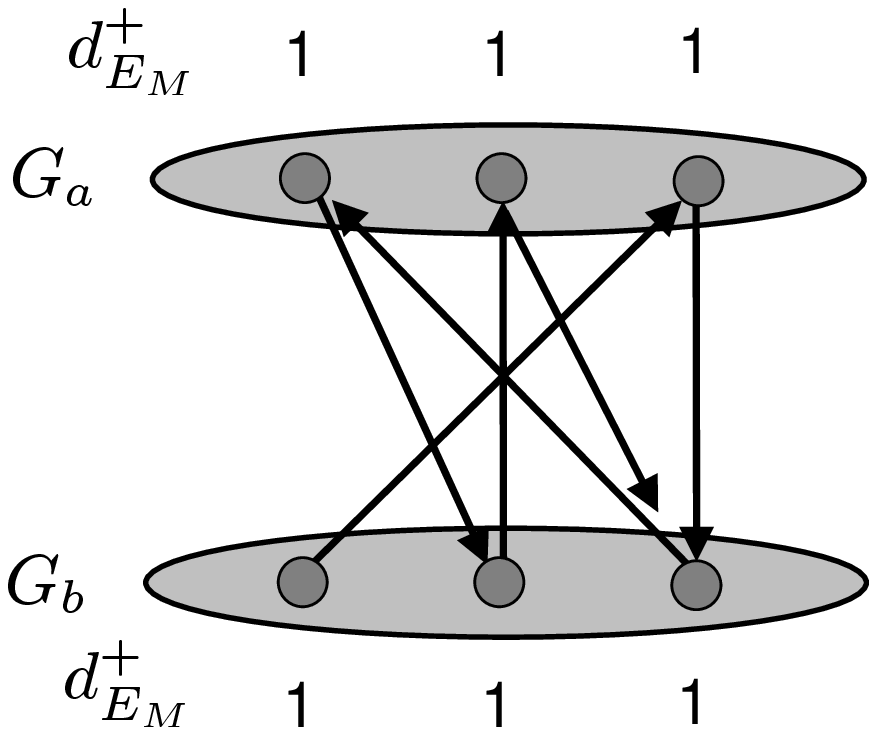}
\includegraphics[scale= 0.4]{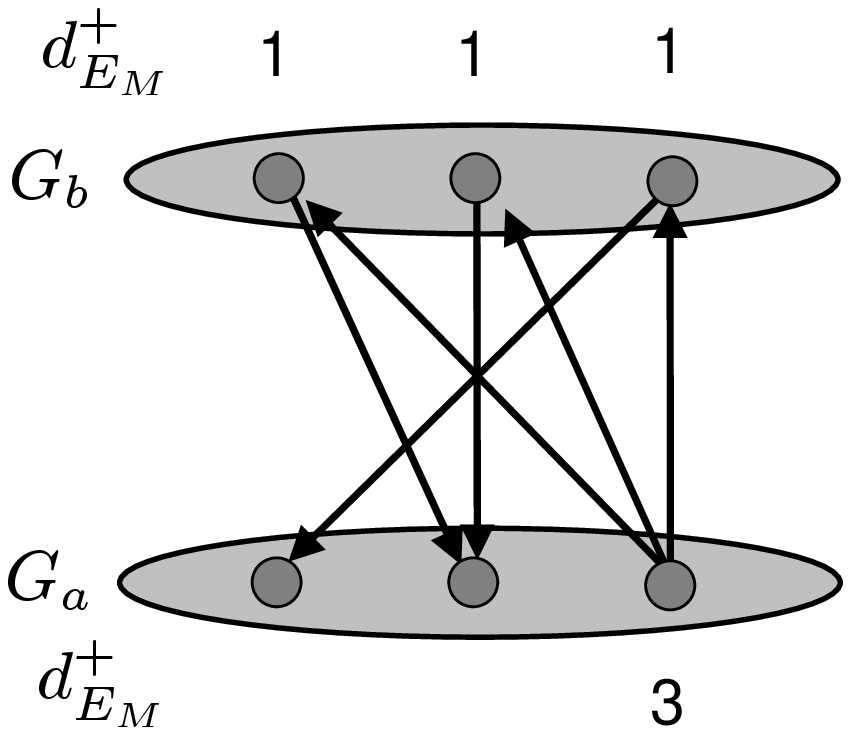}
 \caption{Constructions for the five different DOF allocations
satisfying the hypothesis of Proposition \ref{prop:mergetwo3D}
when each of $G_a$ and $G_b$ has 3 DOFs. The graphs are all rigid
are they have the same undirected underlying graphs as
construction in Fig. \ref{fig:gagb-60-321}, \ref{fig:gagb-60-222}
or \ref{fig:gagb-60-others} or rotated versions of them.}
 \label{fig:gagb-33}
\end{figure}

{We now suppose that at least one of the graphs has less than 3
vertices, and show that a rigid graph can be obtained by adding
directed edges leaving vertices with local DOFs, the number of
these edges being provided in Table \ref{tab:needed_edges}.
Observe that a graph consisting of one single vertex always has 3
DOFs, and thus that it is never needed to use any DOF of the other
graph. Similarly, each vertex of a graph containing two vertices
has at least 2 DOFs, so that at most one DOF of the other graph
needs to be used, and only when the other graph has three or more
vertices. Fig. \ref{fig:merge_3Dsmallgraphs} shows how these
mergings can be performed. Note that the rigidity of the three
first graphs is immediate as they are complete graphs. The
rigidity of the other two follows from the fact that they can be
obtained from $G_b$ by performing one of two operations (v), which
guarantees the rigidity of the graph obtained
\cite{TayWhiteley:85}.

\begin{figure}
\centering
\includegraphics[scale= 0.4]{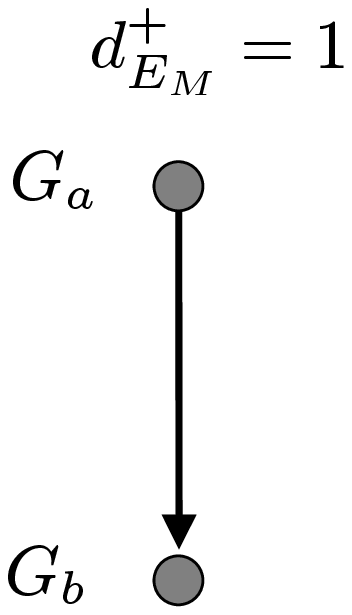}
\includegraphics[scale= 0.4]{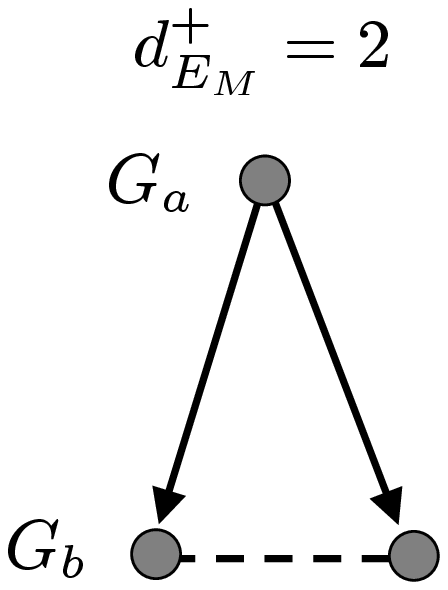}
\includegraphics[scale= 0.4]{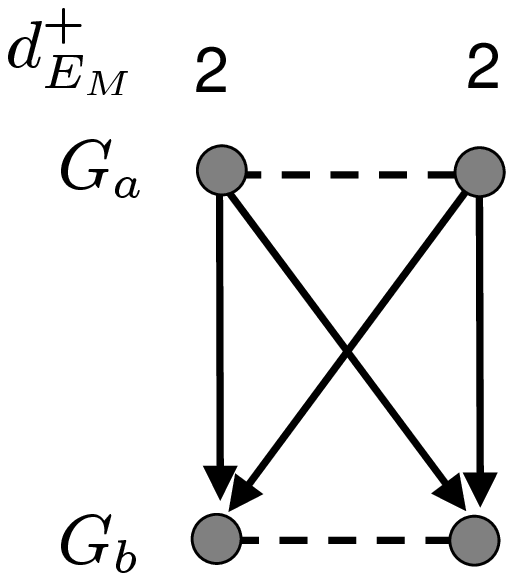}\\
\includegraphics[scale= 0.4]{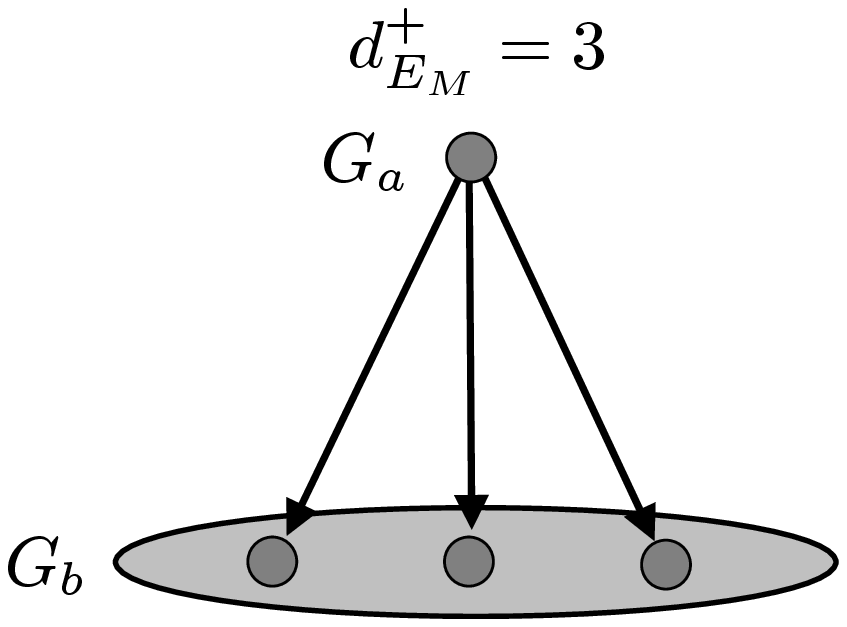}
\includegraphics[scale= 0.4]{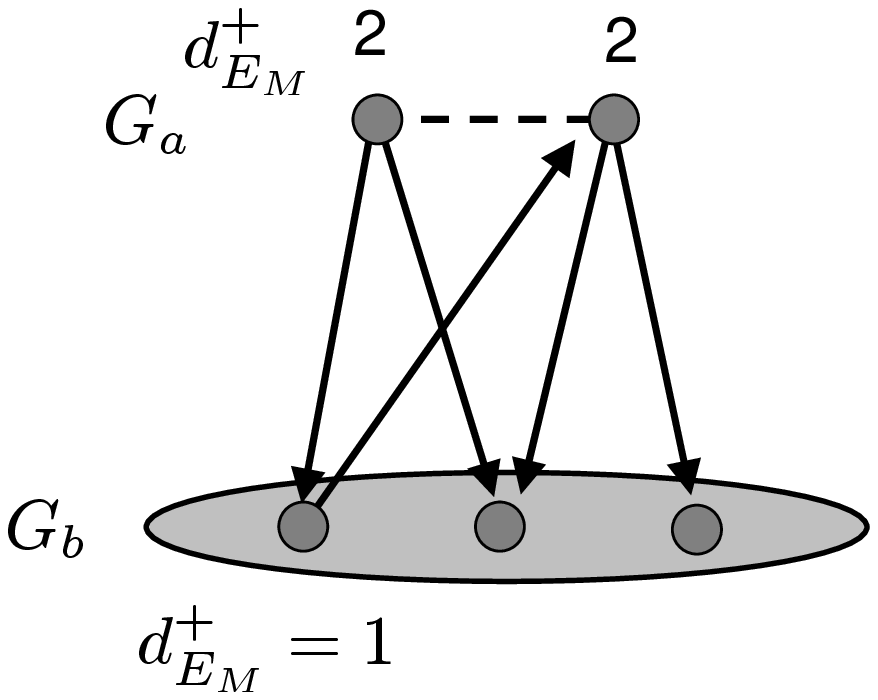}\\
 \caption{Illustration of the merging between
two graphs, one of which at least has less than 3 vertices. The
dashed line represent the internal edge(s) of graphs with two
vertices, the orientation of which is not relevant for our
purpose. The vertex count in $G_b$ is precisely 1,2 and 2 for the
first three and a minimum of 3 for the last
two.}\label{fig:merge_3Dsmallgraphs}
\end{figure}
}
\end{document}